\documentclass[journal]{IEEEtran}
\usepackage{multicol}
\usepackage{amsmath,epsfig,comment}
\usepackage{amsthm}
\usepackage{tcolorbox}
\usepackage{amssymb}
\usepackage{multirow}
\usepackage{mathrsfs}
\usepackage{amsmath}
\usepackage{arydshln}
\usepackage{multirow}
\usepackage{arydshln}

\usepackage{cases} 
\usepackage{amssymb,amsmath,cite}
\usepackage{epsfig}
\usepackage{color}
\usepackage{bm}
\usepackage{graphicx,subfigure}
\usepackage{algorithm}
\usepackage{algorithmic}
\usepackage{multirow}
\usepackage{soul}

\usepackage{extarrows}

\usepackage{geometry}
\def\s{\mathbf{s}}

\def\bX{\mathbf{X}}

\def\RR{\mathcal{R}}

\def\R{\mathbb{R}}

\def\C{\mathbb{C}}

\def\bW{\mathbf{W}}
\def\bV{\mathbf{V}}

\def\bH{\mathbf{H}}
\def\bU{\mathbf{U}}
\def\bQ{\mathbf{Q}}

\def\bP{\mathbf{P}}

\def\y{\mathbf{y}}

\def\x{\mathbf{x}}
\def\s{\mathbf{s}}

\def\h{\mathbf{h}}
\def\bF{\mathbf{F}}
\def\bA{\mathbf{A}}

\def\bZ{\mathbf{Z}}
\def\bD{\mathbf{D}}

\def\bH{\mathbf{H}}

\def\bG{\mathbf{G}}
\def\bh{\mathbf{h}}
\def\bthe{\boldsymbol{\Theta}}

\def\bw{\mathbf{w}}

\def\u{\mathbf{u}}

\def\bw{\mathbf{w}}
\def\bY{\mathbf{Y}}

\newtheorem{theorem}{Theorem}

\newtheorem{lemma}{Lemma}
\newtheorem{remark}{Remark}

\newtheorem{proposition}{Proposition}

\newtheorem{corollary}{Corollary}

\newtheoremstyle{noparens}%
  {}{}%
  {\itshape}{}%
  {\bfseries}{.}%
  { }%
  {\thmname{#1}\thmnumber{ #2}\mdseries\thmnote{ #3}}

\theoremstyle{noparens}

\title{Microwave Linear Analog Computer (MiLAC)-aided Multiuser MISO: Fundamental Limits and Beamforming Design}
\author{\IEEEauthorblockN{Zheyu Wu, Matteo Nerini, and  Bruno Clerckx}
  	\thanks{Z. Wu, M. Nerini, and B. Clerckx are with the Department of Electrical and Electronic Engineering, Imperial College London, London, SW7 2AZ, U.K. (email: \{zheyu.wu, m.nerini20, b.clerckx\}@imperial.ac.uk).}
  }
\date{today}
\begin{document}
\maketitle
\begin{abstract}
As wireless communication systems evolve toward the 6G era, ultra-massive/gigantic multiple-input multiple-output (MIMO)
is envisioned as a key enabling technology. Recently, microwave linear analog computer (MiLAC) has emerged as a promising approach to realize beamforming entirely in the analog domain, thereby alleviating the scalability challenges associated with gigantic MIMO. In this paper, we investigate the fundamental beamforming flexibility and design of lossless and reciprocal MiLAC-aided beamforming for multiuser multiple-input single-output (MISO) systems. We first provide a rigorous characterization of the set of beamforming matrices achievable by MiLAC. Based on this characterization, we prove that MiLAC-aided beamforming does not generally achieve the full flexibility of digital beamforming, while offering greater flexibility than conventional phase-shifter-based analog beamforming. Furthermore, we propose a hybrid digital–MiLAC architecture and show that it achieves digital beamforming flexibility when the number of radio frequency (RF) chains equals the number of data streams, halving that required by conventional hybrid beamforming.
We then formulate the MiLAC-aided sum-rate maximization problem for MU-MISO systems. To solve the problem efficiently,  we reformulate the MiLAC-related constraints as a convex linear matrix inequality and establish a low-dimensional subspace property that significantly reduces the problem dimension. Leveraging these results, we propose weighted minimum mean-square error (WMMSE)-based algorithms for solving the resulting problem.  Simulation results demonstrate that MiLAC-aided beamforming achieves performance close to that of digital beamforming in gigantic MIMO systems. Compared with hybrid beamforming, it achieves comparable or superior performance with lower hardware and computational complexity by avoiding symbol-level digital processing and enabling low-resolution digital-to-analog converters (DACs). 
\end{abstract}

\section{Introduction}

Future sixth-generation (6G) wireless communication systems are expected to evolve toward ubiquitous coverage, extreme data rates, massive connectivity, and ultra-low latency \cite{6G_survey}.  To meet these growing demands, wireless systems are being driven from conventional sub-6 GHz bands toward upper mid-band frequencies (7-24 GHz), where wider bandwidths can be exploited to enhance system capacity \cite{6G_emil}. Such a frequency shift motivates the deployment of extremely large-scale antenna arrays, which are essential for compensating increased path loss and providing sufficient beamforming gain and high spatial multiplexing capability. This gives rise to the concept of ultra-massive multiple-input multiple-output (MIMO)/gigantic MIMO \cite{6G_emil, qualcomm_6g_vision_whitepaper}. In gigantic MIMO systems, the number of transmit antennas can scale to the order of hundreds or even thousands, far exceeding that of conventional massive MIMO deployed in 5G systems \cite{5G_Larsson}.

However, the extreme antenna scaling in gigantic MIMO poses significant practical challenges to conventional fully digital MIMO architectures, in which each antenna element is connected to a dedicated radio-frequency (RF) chain to perform digital beamforming, as illustrated in Fig. \ref{scheme2} (a). Due to power-hungry components in each RF chain, including high-resolution digital-to-analog converters (DACs) and power amplifiers (PAs), the hardware cost and energy consumption become prohibitive as the number of antennas grows.

Various approaches have been proposed to address this issue. One line of work considers employing very low-resolution DACs, e.g., one-bit DACs, at the RF chains. However, the resulting coarse quantization typically leads to severe performance degradation \cite{ZF,SEPlinear}.
 Another widely adopted solution is hybrid digital-analog beamforming \cite{el2014spatially,hybrid1,hybrid2}. In this architecture, a limited number of RF chains are used to implement a low-dimensional digital beamformer, which is followed by an analog beamforming stage realized using phase shifters; see Fig. \ref{scheme2} (b).  Hybrid beamforming can achieve the same beamforming flexibility as fully digital beamforming using significantly fewer RF chains, whose number scales only as twice the total number of data streams and is irrelevant to the number of transmit antennas \cite{hybrid2}. To further reduce the power consumption and improve the energy efficiency of hybrid beamforming, the concept of tri-hybrid MIMO has been recently proposed \cite{triMIMO}. The tri-hybrid architecture combines conventional hybrid beamforming with reconfigurable antenna arrays \cite{shlezinger2021dynamic,fluid,zhu2023movable}
, which serve as a complementary layer to perform beamforming  in the electromagnetic (EM) domain.


 Very recently, the concept of microwave linear analog computer (MiLAC) has been proposed in \cite{part1,part2}. A MiLAC is a multiport microwave network composed of numerous tunable impedance components. By tuning those impedances, the network is fully reconfigurable and can  implement a broad class of linear transformations directly in the analog domain.  For example, it has been shown in \cite{part1} that MiLAC enables the computation of linear minimum mean-square error (LMMSE) operators with a computational complexity that scales quadratically with the problem dimension, rather than cubically as in conventional digital implementations.  Beyond the computational capabilities, MiLACs also exhibit strong potential in wireless communication systems, particularly in gigantic MIMO scenarios. As illustrated in Fig. \ref{scheme2} (c), a MiLAC can be deployed at the transmitter, where its input ports are connected to the RF chains and its output ports are connected to the transmit antennas. By feeding the data symbols into the input ports, the MiLAC performs beamforming entirely in the analog domain as the signals propagate through the microwave network.
\begin{figure}
\includegraphics[width=0.5\textwidth]{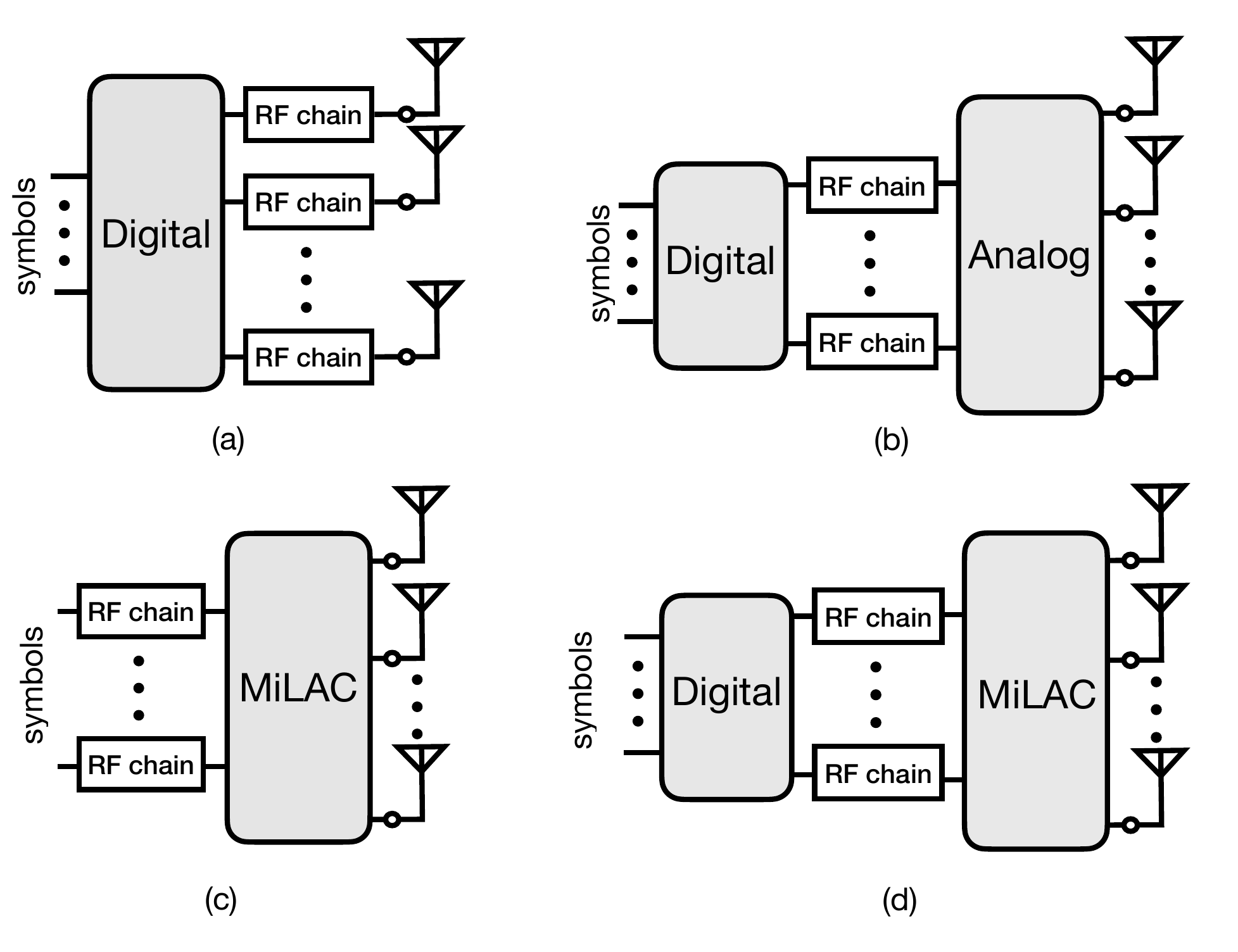}
\centering
\caption{Different beamforming schemes, where (a) digital beamforming, (b) hybrid beamforming,  (c) MiLAC-aided beamforming, and (d) hybrid digital-MiLAC beamforming.}
\label{scheme2}
\end{figure}

As discussed in \cite{part2}, MiLAC-aided beamforming provides significant advantages in terms of hardware cost and computational complexity.  First, the required number of RF chains equals the number of transmitted data symbols, which is  much smaller than the number of antennas.
 Second, it suffices to employ low-resolution DACs at the RF chains, since the RF chains directly carry the constellation symbols. Third, beamforming is performed entirely in the analog domain, thereby eliminating per-symbol digital processing. Fourth, MiLAC-aided beamforming enables the efficient implementation of specific beamforming matrices, such as zero-forcing beamforming, with substantially reduced computational complexity.

The early works on MiLAC \cite{part1,part2} adopt an idealized assumption that the microwave network can be arbitrarily reconfigured. Under this assumption, MiLAC-aided beamforming was shown to achieve the same flexibility as digital beamforming, highlighting its strong potential as a fully analog beamforming architecture. However, in practice, microwave networks are subject to a variety of physical constraints, which makes such idealized assumptions unrealistic. In particular, the microwave network is typically assumed to be lossless and reciprocal. The lossless constraint ensures that the network does not dissipate (or amplify) signal power, while the reciprocity constraint implies symmetric signal transmission between any pair of ports, which facilitates the implementation of microwave networks and avoids the need of non-reciprocal devices. Taking these practical constraints into account, it was shown in \cite{MIMOcapacity} that, for point-to-point MIMO systems, MiLAC-aided beamforming can still achieve the same channel capacity as fully digital beamforming. Furthermore, it was demonstrated in \cite{MIMOcapacity2} that the number of tunable impedance components can be substantially reduced while preserving the capacity-achieving property.

Despite these progress, the fundamental limits and optimization of MiLAC-aided beamforming in multiuser systems remain largely unexplored.  In this paper, we investigate the fundamental flexibility and beamforming design of a lossless and reciprocal MiLAC in multiuser systems. 
The main contributions are summarized as follows.

1) \emph{Characterization of the fundamental flexibility of MiLAC-aided beamforming.} We characterize the fundamental flexibility of MiLAC-aided beamforming by deriving a compact and explicit expression for the set of achievable beamforming matrices. This characterization enables a rigorous comparison between MiLAC-aided beamforming and classical beamforming schemes, e.g.,  digital beamforming and hybrid beamforming, and provides key insights for the design of MiLAC-aided beamforming algorithms.

2) \emph{Rigorous comparisons with classical beamforming schemes.}  We show that MiLAC-aided beamforming does not, in general, attain the full flexibility of digital beamforming; in particular, it cannot realize arbitrary beamforming matrices with highly correlated beamforming vectors. However, in gigantic MIMO systems,  the number of antennas is extremely large and user channels become asymptotically orthogonal, thus MiLAC-aided beamforming achieves performance comparable to that of digital beamforming. Furthermore, we show that MiLAC-based analog beamforming provides greater flexibility than conventional phase-shifter-based analog beamforming, indicating that replacing the phase-shifter-based analog part in classical hybrid beamforming architectures with a MiLAC leads to enhanced design flexibility.

3) \emph{A digital-achieving hybrid digital-MiLAC beamforming}. We show that a hybrid digital–MiLAC beamforming scheme, in which the analog beamforming stage is implemented using a MiLAC as illustrated in Fig. \ref{scheme2} (d), achieves the same flexibility as digital beamforming with only  $K$ RF chains, where $K$ is the number of served users. This  reduces the RF-chain requirement compared to classical hybrid beamforming architectures, which require $2K$ RF chains to achieve full flexibility \cite{hybrid2}.

4) \emph{Efficient MiLAC-aided beamforming strategies}. We formulate the MiLAC-aided sum-rate maximization problem for multiuser multi-input single-output (MU-MISO) systems. We show that the set of MiLAC-aided beamforming matrices can be equivalently characterized by a convex linear matrix inequality (LMI). In addition, we establish a desirable low-dimensional subspace property of the formulated problem, substantially reducing the problem dimension.  Building on this structure, we propose a weighted minimum mean-square error (WMMSE) algorithm, where the resulting beamforming subproblem takes the form of a quadratic semidefinite program. To further reduce the computational cost, we propose a low-complexity algorithm that admits closed-form updates and eliminates the need to handle any positive semidefinite constraints.

Simulation results demonstrate that MiLAC-aided beamforming achieves performance close to fully digital beamforming, and achieves comparable or superior performance compared to conventional hybrid beamforming in multiuser systems. In particular, MiLAC-aided beamforming exhibits more pronounced performance advantages in the low-SNR regime and with a large number of transmit antennas (which is the case of gigantic MIMO). Beyond performance, MiLAC-aided beamforming avoids symbol-level digital processing, such as per-symbol matrix–vector multiplications, and enables the use of low-resolution DACs, thereby achieving reduced hardware and computational complexity compared with hybrid beamforming.

During the final preparation of this manuscript, a preprint \cite{fang2026} appeared, which independently studies MiLAC-aided beamforming in multiuser scenarios and arrives at similar conclusions regarding its limited flexibility compared to digital beamforming (using a different approach). Beyond  \cite{fang2026}, our work further provides: 1) a compact characterization of the achievable beamforming matrix set; 2) comparisons with phase-shifter-based analog beamforming and hybrid beamforming; 3) analysis of a hybrid digital–MiLAC scheme; and 4) efficient algorithmic designs leveraging the derived set and the special problem structure.

\emph{Organization}: The rest of the paper is organized as follows. In Section~\ref{sec:model}, we introduce the MiLAC-aided system model and formulate the corresponding beamforming design problem. Section~\ref{sec:milac} characterizes the flexibility of MiLAC-aided beamforming and compares it with classical beamforming schemes. In Section~\ref{sec:4}, we develop efficient algorithms to solve the MiLAC-aided beamforming design problem. Section~\ref{sec:simulation} presents simulation results to demonstrate the effectiveness of the proposed algorithms and provides a comprehensive comparison of fully digital, MiLAC-aided, and hybrid beamforming schemes. Finally, Section~\ref{sec:conclusion} concludes the paper.

\emph{Notations:}
Throughout the paper, $\mathbb{C}$ and $\mathbb{R}$ denote the complex space and the real space, respectively. Scalars, column vectors, matrices, and sets are denoted by
 $x$, $\x$, $\mathbf{X}$, and $\mathcal{X}$, respectively. For a matrix $\mathbf{X}$,  $[\mathbf{X}]_{\mathcal{I}_1,\mathcal{I}_2}$ denotes its submatrix formed by the rows indexed by $\mathcal{I}_1$ and the columns indexed by $\mathcal{I}_2$. For simplicity, $X_{i,j}$ is also used to denote the $(i,j)$-th element of $\bX$. In particular, for a vector $\x$, $x_k$ denotes its $k$-th element. The notation $\bX^{\frac{1}{2}}$ represents the square root of a positive semidefinite matrix $\bX$, and $\text{tr}(\bX)$ denotes the trace of $\bX$. For a vector $\x$, $\text{diag}(\x)$ returns a diagonal matrix whose diagonal entries are $\x$.
 The operators $(\cdot)^T$, $(\cdot)^H$, $(\cdot)^{-1}$, and $\RR(\cdot)$ return the transpose, the Hermitian transpose, the inverse, and the real part of their corresponding argument, respectively. The notation  $\|\cdot\|_2$  refers to the  $\ell_2$-norm of a vector or the spectral norm of a matrix, and $\|\cdot\|_F$ denotes the Frobenius norm of a matrix.  Given two Hermitian matrices $\bX$ and $\bY$, $\bX\succeq\bY$ means that $\bX-\bY$ is positive semidefinite; in particular, $\bX\succeq \mathbf{0}$ indicates that $\bX$ is positive semidefinite. The inner product of two matrices $\bX$ and $\bY$ is defined as  $\left<\bX,\bY\right>=\RR(\text{tr}(\bX^H\bY))$. The operator $\circ$ denotes the Hadamard (element-wise) product.   The symbols $\mathbf{I}$ and $\mathbf{0}$  refer to the identity matrix and the all-zero matrix, respectively, with dimensions specified by subscripts when necessary.
\section{System Model and Problem Formulation}\label{sec:model}
 \begin{figure}
\includegraphics[width=0.5\textwidth]{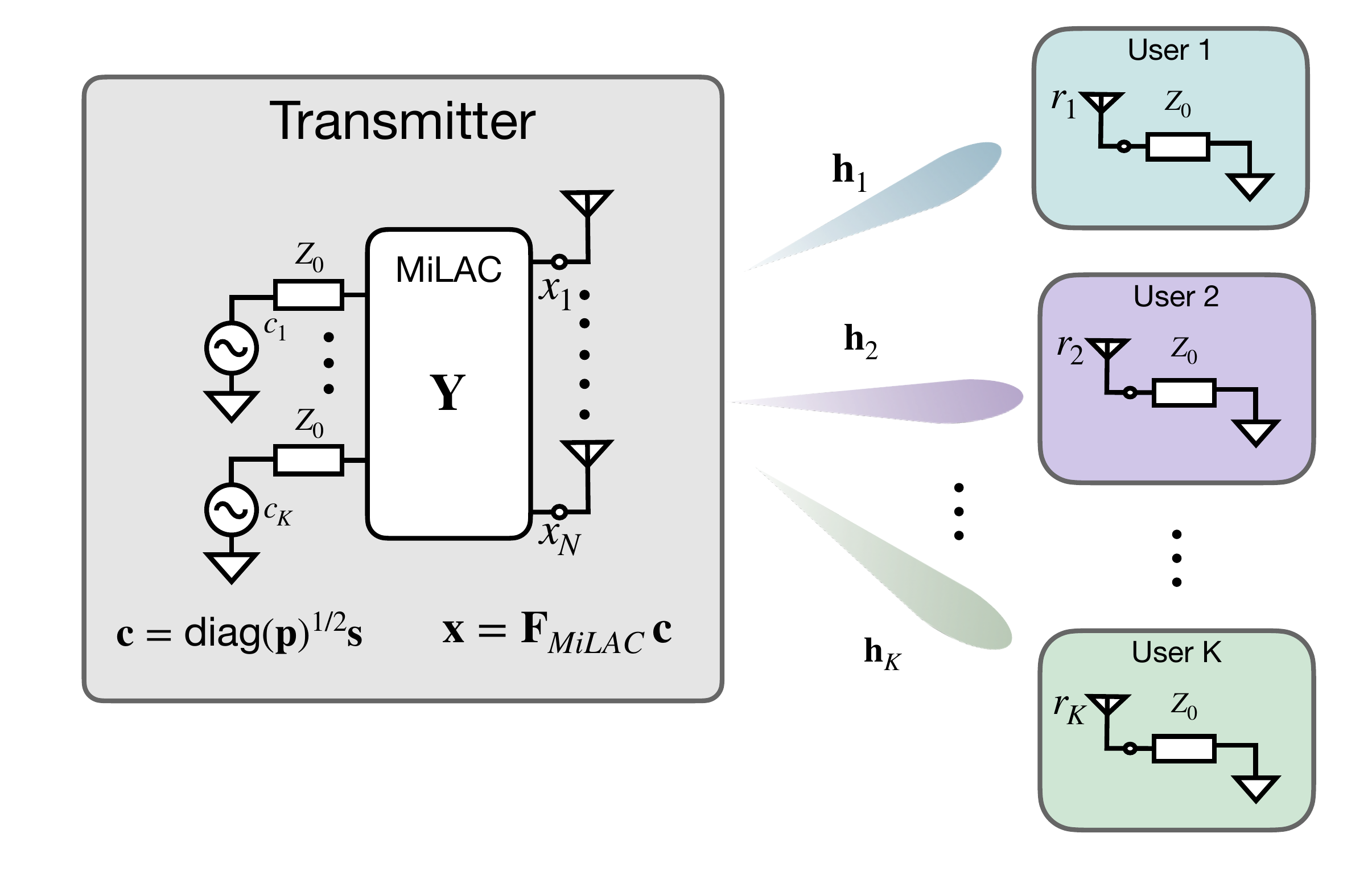}
\centering
\vspace{-0.8cm}
\caption{A MiLAC-aided multiuser MISO system.}
\label{fig:model}
\end{figure}
\subsection{System Model}
Consider a MiLAC-aided multiuser MISO system, where an $N$-antenna transmitter severs $K$ single-antenna users with the help of a MiLAC, as shown in Fig. \ref{fig:model}. The transmitter-side hardware architecture follows that in \cite{part1,part2,MIMOcapacity,MIMOcapacity2}, which consists of $K$ RF chains used to transmit $K$ information symbols in parallel and a $(K\times N)$-port MiLAC that further precodes the signal. Each RF chain is modeled as a voltage source with a series impedance of $Z_0=50\,\Omega$, and the MiLAC is modeled as a linear microwave network comprising multiple tunable admittance (impedance) components; see Fig. \ref{milac_arch} for an example of a four-port MiLAC.

 Let $\mathbf{s}=[s_1,s_2,\dots,s_K]^T$ denote the symbol vector for the users, where $\mathbb{E}[\s\s^H]=\mathbf{I}_K$ and  $s_k$ is the symbol for user $k$, and let $\mathbf{p}=[p_1,p_2,\dots,p_K]\in\R^{K}$ denote the power allocation vector, where $p_k\geq 0$ is the power allocated to symbol $s_k$.    The source signal at the RF chains is given by 
$\mathbf{c}=\text{diag}(\mathbf{p})^{\frac{1}{2}}\mathbf{s}.$ The signal $\mathbf{c}$ is then fed into the MiLAC, and we denote the output  signal by $\x=\mathbf{F}_{\text{MiLAC}}\,\mathbf{c}.$  Here, $\mathbf{F}_{\text{MiLAC}}\in\C^{N\times K}$ captures   the input–output relationship of the MiLAC. We relegate a detailed discussion to Section \ref{sec:milacmodeling}.

The receive signal at the $k$-th user then reads
$$r_k=\h_k^H\mathbf{F}_{\text{MiLAC}}\,\text{diag}(\mathbf{p})^{\frac{1}{2}}\s+n_k,$$
where $\h_k\in\C^{N}$ is the channel between the transmitter and the $k$-th user, and $n_k\sim\mathcal{CN}(0,\sigma_n^2)$ is the additive white Gaussian noise (AWGN).
\vspace{-0.3cm}
\subsection{MiLAC Modeling}\label{sec:milacmodeling}
 In this paper, we focus on a fully-connected MiLAC architecture as in \cite{part1,part2,MIMOcapacity}, where each port is connected to a tunable admittance to ground, and every pair of ports is interconnected via a tunable admittance.

   \begin{figure}
\includegraphics[width=0.3\textwidth]{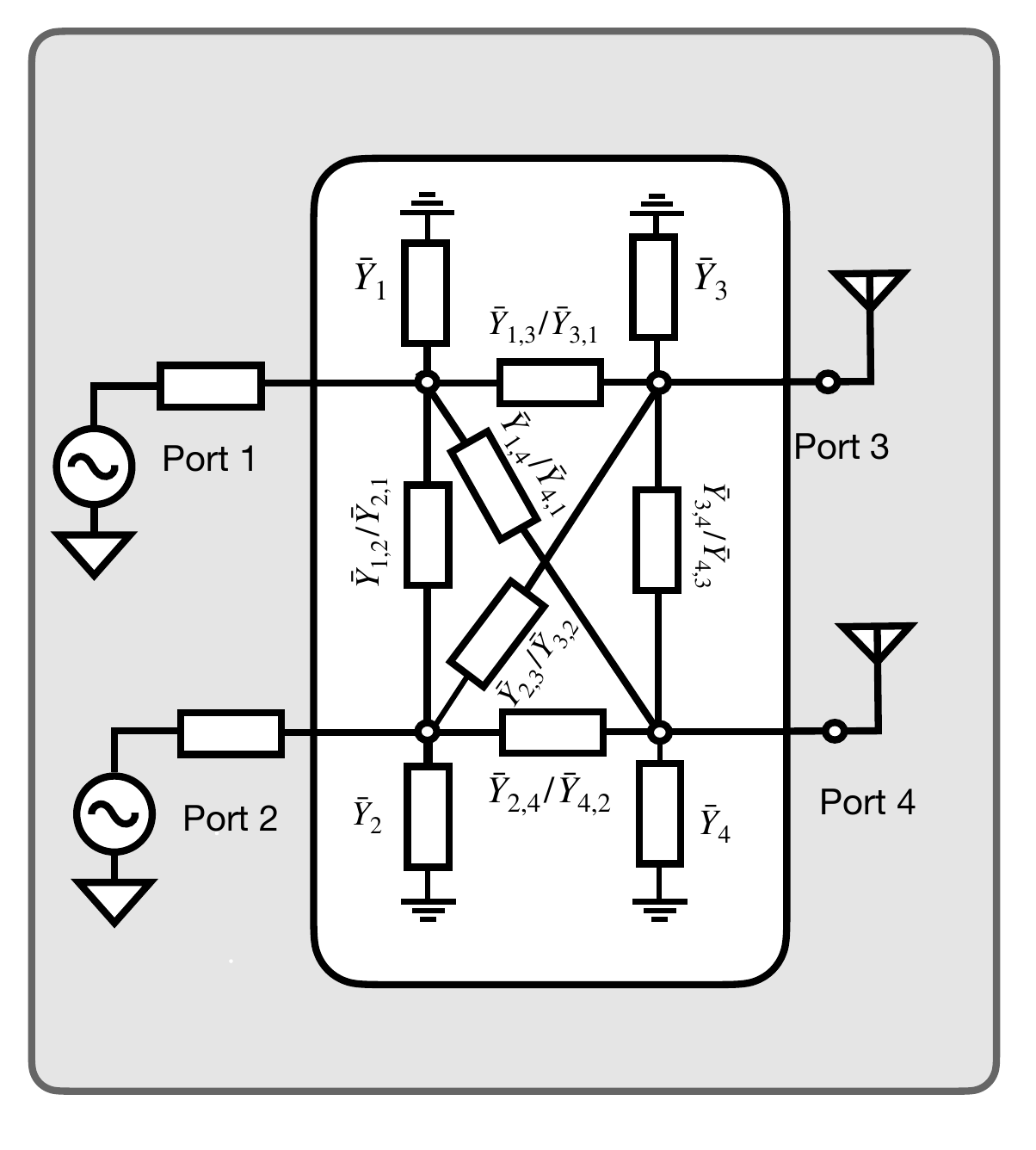}
\centering
\vspace{-0.3cm}
\caption{An example of a four-port MiLAC with $N=K=2$.}
\label{milac_arch}
\end{figure}
 According to \cite{part1,part2}, $\bF_{\text{MiLAC}}$ is a function of the MiLAC admittance matrix $\bY\in\C^{(N+K)\times (N+K)}$, given by
\begin{equation}\label{Fmilac:Y}
\mathbf{F}_{\text{MiLAC}}=\left[\left(\mathbf{I}_{N+K}+{Z_0\mathbf{Y}}\right)^{-1}\right]_{K+1:K+N,1:K}.
\end{equation}
The admittance matrix $\bY$ is  reconfigurable and can be controlled by tuning the admittances in the network. Let $\bar{Y}_n$ be the admittance connecting  the $n$-th port of the network to ground, and let $\bar{Y}_{n,j}$ denote the admittance connecting the $n$-th and $j$-th ports.  The admittance $\bY$, according to multiport network theory \cite{microwavebook},   can be expressed as 
\begin{equation}\label{admittance}
[\bY]_{n,m}=\left\{
\begin{aligned}
&-\bar{Y}_{n,m},~~~~~~~~~\,~~&\text{if}~n\neq m;\\
&\bar{Y_n}+\textstyle\sum_{j\neq n}\bar{Y}_{n,j},~~&\text{if }n=m. 
\end{aligned}\right.
\end{equation}
 
  In addition to the admittance matrix $\bY$, the MiLAC can also be characterized by the scattering matrix $\bthe$ of the microwave network \cite{microwavebook}, which relates to $\bY$ through 
  \begin{equation}\label{Eq:theta_Y}
  \bthe=(\mathbf{I}_{N+K}+Z_0\bY)^{-1}(\mathbf{I}_{N+K}-Z_0\bY).
  \end{equation}As shown in \cite{MIMOcapacity,MIMOcapacity2}, expressing $\mathbf{F}_{\text{MiLAC}}$ in terms of $\bthe$ yields a clearer formulation as follows 
   \begin{equation}\label{F}
\bF_{\text{MiLAC}}=\frac{1}{2}[\bthe]_{K+1:K+N,1:K},
\end{equation}
which is obtained by rewriting \eqref{Eq:theta_Y} as $\bthe=-\mathbf{I}_{N+K}+2(\mathbf{I}_{N+K}+Z_0\bY)^{-1}$ and using \eqref{Fmilac:Y}.

\subsubsection{MiLAC without practical constraints}  The works \cite{part1,part2} consider a MiLAC without physical constraints, i.e.,  the  tunable admittances $\{\bar{Y}_i\}$ and $\{\bar{Y}_{i,j}\}_{i\neq j}$ are assumed to take arbitrary complex values.  Under this assumption,  no constraints are imposed on the admittance $\bY$, or equivalently on the scattering matrix
$\bthe$, due to the relationship in \eqref{admittance} and \eqref{Eq:theta_Y}. Consequently, the MiLAC can implement any beamforming matrix \cite{part2}. 

However, the above idealized assumption raises several practical concerns. First, admittance components with non-zero real parts correspond to either loss or gain; loss leads to power dissipation, whereas gain requires active components and incurs additional power consumption. Second, admittance components in microwave networks are typically reciprocal. In light of these considerations, it is more practical to assume that the admittance components in MiLAC are lossless and reciprocal \cite{MIMOcapacity,MIMOcapacity2}.
 \subsubsection{Lossless and reciprocal MiLAC}
In a lossless and reciprocal MiLAC, all tunable admittances $\{\bar{Y}_i\}$ and $\{\bar{Y}_{i,j}\}$ are  purely imaginary and satisfy $\bar{Y}_{i,j}=\bar{Y}_{j,i}$ for $i\neq j$. Consequently, the admittance matrix $\bY$ is a purely imaginary and symmetric matrix. Equivalently, according to microwave network theory, the scattering matrix  $\bthe$ is unitary and symmetric, i.e., 
\begin{equation}\label{theta}
\bthe=\bthe^T,~\bthe^H\bthe=\mathbf{I}_{N+K}.
\end{equation}  Throughout the paper, we focus on a lossless and reciprocal MiLAC. 
Unless otherwise specified, the term \emph{MiLAC} in the remainder of the paper refers to a lossless and reciprocal MiLAC.


\subsection{Problem Formulation}
 We refer to the beamforming schemes supported by the architecture in Fig. \ref{fig:model} with a lossless and reciprocal MiLAC as MiLAC-aided beamforming, and define the corresponding MiLAC-aided beamforming matrix set,
 under a maximum transmit power  $P_T$,  as:
\begin{equation}\label{Wmilac0}
\begin{aligned}
\mathcal{W}_{\text{MiLAC}}=\big\{\bW&=\mathbf{F}\text{diag}(\mathbf{p})^{\frac{1}{2}}\mid\bF=[\bthe]_{K+1:K+N,1:K},\\
&\bthe=\bthe^T,~\bthe^H\bthe=\mathbf{I}_{N+K},~\mathbf{1}^T\mathbf{p}\leq P_T\big\},
\end{aligned}
\end{equation}
where $\bF:=2\bF_{\text{MiLAC}}$ is introduced for notational simplicity\footnote{Rigorously, the set should be defined with $\mathbf{F}_{\text{MiLAC}}$ rather than $\mathbf{F}$, or equivalently,  there should be a factor of ${1}/{2}$ before $\bF$ in the definition of $\mathcal{W}_{\text{MiLAC}}$. For notational simplicity, we omit this factor and absorb it into the channel normalization without loss of generality (see \eqref{sumrate}).}. 
It has been shown in \cite{MIMOcapacity} that $\mathcal{W}_{\text{MiLAC}}$ contains all matrices with transmit power $P_T$ and unitary columns, i.e., 
\begin{equation}\label{unitary}
\{\bW\mid\|\bW\|_F^2=P_T, \bw_i^H\bw_j=0,~\forall~i\neq j\}\subseteq \mathcal{W}_{\text{MiLAC}}.
\end{equation}
Thus, MiLAC is capacity-achieving for point-to-point MIMO systems \cite{MIMOcapacity}. 

Our goal  in this paper is  twofold. First, we aim to provide a complete characterization of $\mathcal{W}_{\text{MiLAC}}$ to gain insights into the class of beamforming matrices that a MiLAC can realize. In particular, we seek to answer the following question: In general multiuser systems, can MiLAC-aided beamforming still achieve the same flexibility as fully digital beamforming? Second, we aim to develop efficient and effective beamforming strategies for MiLAC-aided multiuser MISO systems and to numerically compare their performance against fully digital and hybrid beamforming. Specifically, we focus on the following sum-rate maximization problem:
\begin{equation}\label{sumrate}
\begin{aligned}
\max_{\bW}~&\sum_{k=1}^K\log\left(1+\frac{\frac{1}{4}|\h_k^H\bw_k|^2}{\frac{1}{4}\sum_{j\neq k}|\h_k^H\bw_j|^2+\sigma_n^2}\right)\\
\text{s.t. }~&\bW\in\mathcal{W}_{\text{MiLAC}},
\end{aligned}
\end{equation}
where the factor $\frac{1}{4}$ arises because $\bF=2\bF_{\text{MiLAC}}$.
\section{Characterization of MiLAC-Aided Beamforming and Comparison with Classical Beamforming Schemes}\label{sec:milac}
In this section, we characterize the flexibility of MiLAC-aided beamforming. We first express the class of beamforming matrices achievable by MiLAC-aided beamforming, i.e., $\mathcal{W}_{\text{MiLAC}}$, in a concise and explicit form in Section \ref{sec:3a}. We then compare MiLAC-aided beamforming with fully digital beamforming and phase-shifter-based analog beamforming in Sections \ref{sec:3b} and \ref{sec:3c}, respectively. In particular, we show that MiLAC-aided beamforming generally cannot achieve the full flexibility of fully digital beamforming, but it offers greater flexibility than phase-shifter-based analog beamforming. Finally, in Section \ref{sec:3d}, we discuss a hybrid digital–MiLAC architecture, which provides more flexibility than traditional hybrid beamforming  and  is digital-achieving with $K$ RF chains.
\subsection{Beamforming Matrices of MiLAC-Aided Beamforming}\label{sec:3a}
According to \eqref{F}, the MiLAC affects its output signal $\x$ only though the submatrix $[\bthe]_{K+1:K+N,1:K}$. In the following proposition, 
we translate the unitary and symmetric constraints \eqref{theta} on $\bthe$  into the corresponding constraints on the submatrix  $[\bthe]_{K+1:K+N,1:K}$.
 \begin{proposition}\label{pro1}
First, let $\bthe\in\C^{(N+K)\times (N+K)}$ be a unitary and symmetric matrix, then $\mathbf{F}:=[\bthe]_{K+1:K+N,1:K}$ satisfies  $$\|\mathbf{F}\|_2\leq 1.$$ Second, for any matrix $\mathbf{F}\in\C^{N\times K}$ satisfying $\|\bF\|_2\leq 1$, there exists matrices $\bthe_{11}\in\C^{N\times N}$ and $\bthe_{22}\in\C^{K\times K}$ such that 
\begin{equation}\label{construct:theta}
\bthe=\left[\begin{matrix}\bthe_{11}&\bF^T\\\bF&\bthe_{22}\end{matrix}\right]
\end{equation} is unitary and symmetric.\end{proposition}
\begin{proof}
To prove the first result, we notice that given a unitary matrix $\bthe\in\C^{(N+K)\times (N+K)}$, we have $$\left\|\bF\right\|_2=\left\|\left[\begin{matrix}\mathbf{0}_{N\times K}&\mathbf{I}_{N}\end{matrix}\right]\bthe\left[\begin{matrix}\mathbf{I}_{K}\\\mathbf{0}_{N\times K}\end{matrix}\right]\right\|_2\leq \|\bthe\|_2=1.$$

To prove the second result, we notice that given any $\|\bF\|_2\leq 1$, let $\bF=\bU_{F}\bD_{F}\bV_{F}^H$ be the singular value decomposition (SVD) of $\bF$, where $\bU_F\in\mathcal{U}(N)$, $\bV_F\in\mathcal{U}(K)$,  $\bD_F=[\text{diag}(\sigma_1,\sigma_2,\dots,\sigma_r,\mathbf{0}_{(K-r)\times 1}),\mathbf{0}_{K\times (N-K)}]^T\in\C^{N\times K}$, and $r\leq K$ is the rank of $\bF$.
We construct $\bthe_{11}$ and $\bthe_{22}$ as
$$
\bthe_{11}=-\bV_F^*\left[\begin{smallmatrix}\sqrt{1-\sigma_1^2}&&&&\\&\ddots&&&\\&&\sqrt{1-\sigma_r^2}&&\\&&&\mathbf{I}_{K-r}\end{smallmatrix}\right]\bV_F^H$$
and 
$$
\bthe_{22}=\bU_F\left[\begin{smallmatrix}\sqrt{1-\sigma_1^2}&&&\\&\ddots&&\\&&\sqrt{1-\sigma_r^2}&\\&&&\mathbf{I}_{N-r}\end{smallmatrix}\right]\bU_F^T.$$
Obviously, $\bthe_{11}$ and $\bthe_{22}$ are symmetric, and thus $\bthe$ defined in \eqref{construct:theta} is symmetric. Next we prove that $\bthe$ is unitary. Since $\bthe$ is symmetric, it suffices to show that 
\begin{subequations}
\begin{align}
&\bthe_{11}\bthe_{11}^H+(\bF^H\bF)^T=\mathbf{I}_{K},\\
&\bF\bF^H+\bthe_{22}\bthe_{22}^H=\mathbf{I}_{N},\\
&\bF\bthe_{11}^H+\bthe_{22}\bF^*=\mathbf{0}_{N\times K},
\end{align}
\end{subequations}
which are obtained from $\bthe\bthe^H=\mathbf{I}_{N+K}$.
It is straightforward to check that the above conditions are satisfied with the constructed $\bthe_{11}$ and $\bthe_{22}$.
\end{proof}
With Proposition \ref{pro1}, we can equivalently express $\mathcal{W}_{\text{MiLAC}}$ in \eqref{Wmilac0} in a much simpler form as follows. 
\begin{corollary}\label{corollary:wmilac}
The lossless and reciprocal  MiLAC-aided beamforming matrix set defined in \eqref{Wmilac0} can be equivalently expressed as  
\begin{equation}\label{MiLAC2}
\begin{aligned}
\mathcal{W}_{\text{MiLAC}}=\big\{\bW&=\mathbf{F}\text{diag}(\mathbf{p})^{\frac{1}{2}}\mid \|\bF\|_2\leq 1,~\mathbf{1}^T\mathbf{p}\leq P_T\big\}.
\end{aligned}
\end{equation}
\end{corollary}
Compared to \eqref{Wmilac0}, Corollary \ref{corollary:wmilac} gives a much clearer characterization of the beamforming matrices achieved by MiLAC-aided beamforming. 
 It also generalizes the result in \eqref{unitary}, which was established in \cite{MIMOcapacity}. In particular, for  any matrix $\bW=[\bw_1,\bw_2,\dots,\bw_K]$ satisfying $\bw_i^H\bw_j=0$ for all $~i\neq j$ and $\|\bW\|_F^2=P_T$,  let $$\mathbf{f}_{i}=\frac{\bw_i}{\|\bw_i\|_2}\text{  and  }~p_i=\|\bw_i\|_2^2,~i=1,2,\dots, K.$$ Then $\bW=\mathbf{F}\text{diag}(\mathbf{p})^{\frac{1}{2}}$, where $\bF=[\mathbf{f}_1,\mathbf{f}_2,\dots,\mathbf{f}_K]$ satisfies  $\|\bF\|_2=1$ and $\mathbf{p}$ satisfies  $\mathbf{1}^T\mathbf{p}=P_T$, and thus $\bW\in\mathcal{W}_{\text{MiLAC}}$ according to \eqref{MiLAC2}.

\subsection{Comparison with Digital Beamforming}\label{sec:3b}
With \eqref{MiLAC2}, we can now compare the beamforming matrix sets of MiLAC-aided and digital beamforming.

For traditional digital beamforming, the transmitter is equipped with $N$ RF chains, and the source signal at the RF chains is $\mathbf{c}=\mathbf{W}\mathbf{s}$, where $\mathbf{W}\in\C^{N\times K}$ is the digital beamforming matrix. Let $P_T$ be the maximum transmit power at the voltage generators, then $\bW$ must satisfy the power constraint $\|\bW\|_F^2\leq P_T$. As discussed in \cite{MIMOcapacity},  when the transmit antennas are perfectly matched to the source generators series impedance $Z_0$, the signal at the transmit antennas is $\x=\mathbf{c}/2$. Hence, the receive signal at user $k$ is\footnote{Compared to the conventional formulation, the receive signal here contains an additional factor of $1/2$. This arises because the maximum transmit power constraint is imposed on the signal at the voltage generators, $\mathbf{c}$, rather than on the signal at the transmit antennas, $\x$. } 
$$r_k=\frac{1}{2}\mathbf{h}_{k}^H\bW\mathbf{s}+n_k.$$
Analogous to the case of MiLAC-aided beamforming, the digital beamforming matrix set can be defined as follows
$$\mathcal{W}_{\text{digital}}=\{\bW\mid \|\bW\|_F^2\leq P_T\}.$$
The following proposition characterizes the relation between $\mathcal{W}_{\text{MiLAC}}$ and $\mathcal{W}_{\text{digital}}$.
\begin{proposition}\label{proposition2}
The following results hold for $\mathcal{W}_{\text{MiLAC}}$ given by \eqref{MiLAC2} and $\mathcal{W}_{\text{digital}}=\{\bW\mid \|\bW\|_F^2\leq P_T\}$.
\begin{itemize}
\item[(i)] $\mathcal{W}_{\text{MiLAC}}\subseteq\mathcal{W}_{\text{digital}}$;
\item[(ii)] $\mathcal{W}_{\text{MiLAC}}\subsetneqq\mathcal{W}_{\text{digital}}$. In particular, given $\|\bW\|_F^2=P_T$, $\bW\in\mathcal{W}_{\text{MiLAC}}$ only if $\bw_i^H\bw_j=0,~\forall~i\neq j$. 
\end{itemize}
\end{proposition}
\begin{proof}
The proof for (i) is straightforward: 
\begin{equation}\label{Wmilac}
\|\bF\text{diag}(\mathbf{p})^{\frac{1}{2}}\|_F^2\leq \|\bF\|_2^2\|\text{diag}(\mathbf{p})^{\frac{1}{2}}\|_F^2\leq P_T.
\end{equation}
We next prove (ii). Given $\bW\in\mathcal{W}_{\text{MiLAC}}$ with $\|\bW\|_F^2=P_T$, we have 
\begin{equation}\label{proof:eq1}
\bW^H\bW=\text{diag}(\mathbf{p})^{\frac{1}{2}}\bF^H\bF\text{diag}(\mathbf{p})^{\frac{1}{2}}\preceq \text{diag}(\mathbf{p}),
\end{equation}
where the last inequality holds since  $\|\bF\|_2\leq 1$ is equivalent to $\bF^H\bF\preceq\mathbf{I}_K$. 
Using the property that any principal submatrix of a positive semidefinite matrix is also positive semidefinite, we have 
$$\left[\begin{matrix}p_i&0\\0&p_j\end{matrix}\right]-\left[\begin{matrix} \|\bw_i\|^2&\bw_i^H\bw_j\\\bw_j^H\bw_i&\|\bw_j\|^2\end{matrix}\right]\succeq \mathbf{0},~\forall~i<j,$$
where the matrix on the left-hand side is the $2\times 2$ principal submatrix of $\text{diag}(\mathbf{p})-\bW^H\bW$ formed by rows and columns $i$ and $j$.
 It follows that 
 \begin{subequations}
 \begin{align}
& \|\bw_i\|^2\leq p_i, ~\|\bw_j\|^2\leq p_j,\label{proof:eq2a}\\
 &|\bw_i^H\bw_j|^2\leq (p_i-\|\bw_i\|^2)(p_j-\|\bw_j\|^2),~\forall~i<j.\label{proof:eq2b}
 \end{align}
 \end{subequations}
We claim that \eqref{proof:eq2a} holds with equality, i.e., $\|\bw_i\|^2=p_i$ for all $i=1,2,\dots, K;$ otherwise $\|\bW\|_F^2<\mathbf{1}^T\mathbf{p}=P_T$, which  contradicts with the condition $\|\bW\|_F^2=P_T$. It then follows from \eqref{proof:eq2b} that 
  $\bw_i^H\bw_j=0$ for all $i\neq j$. 
\end{proof}
Proposition \ref{proposition2} shows that digital beamforming offers greater flexibility in beamforming matrix design than MiLAC-aided beamforming. Specifically, every beamforming matrix achieved by MiLAC-aided beamforming can also be achieved using digital beamforming (Proposition \ref{proposition2} (i)), whereas the converse does not hold: there exist beamforming matrices achieved with digital beamforming that MiLAC-aided beamforming cannot realize (Proposition \ref{proposition2} (ii)). A similar conclusion has also been established in \cite{fang2026}. However, \cite{fang2026} studies the set  $\mathcal{W}_{\text{MiLAC}}$ in its original form in \eqref{Wmilac0}, rather than expressing it in the more compact form given in \eqref{MiLAC2}. This compact formulation further enables systematic comparisons with other beamforming schemes and facilitates efficient algorithm design, which will be detailed in the following (sub)sections. 

\begin{remark}
The conclusion in Proposition \ref{proposition2} differs from that in \cite{part2}, where an unconstrained MiLAC was shown to achieve the same flexibility as digital beamforming. This arises because the lossless and reciprocity constraints restrict the flexibility of  \(\mathbf{F}\). Interestingly, from the proof of Proposition \ref{pro1}, it can be readily shown that relaxing the lossless and reciprocity constraints \(\boldsymbol{\Theta}^H \boldsymbol{\Theta} = \mathbf{I}_{N+K}\) and \(\boldsymbol{\Theta} = \boldsymbol{\Theta}^T\) to the passivity constraint \(\boldsymbol{\Theta}^H \boldsymbol{\Theta} \preceq \mathbf{I}_{N+K}\) leads to the same achievable beamforming matrix set.  
That is, the set \(\mathcal{W}_{\text{MiLAC}}\) given in Corollary \ref{corollary:wmilac} remains unchanged even when the MiLAC is subject only to a passivity constraint. This observation indicates that the reduced flexibility of MiLAC-aided beamforming is an inherent consequence of passive  microwave networks, whereas imposing lossless and reciprocity constraints does not further reduce the achievable flexibility. 

For clarity, we give a summary of physical constraints on MiLAC and the corresponding achievable beamforming matrices in Table \ref{summary}. \end{remark}

\begin{table*}
\fontsize{7}{7}\selectfont{
\caption{A summary of physical constraints on MiLAC and the achievable beamforming matrix.}
\begin{tabular}{c|c|c|c}
\hline
\multirow{2}{*}{Physical Constraints on MiLAC}&\multirow{2}{*}{Mathematical Representation}&\multirow{2}{*}{Achievable Beamforming Matrix}&\multirow{2}{*}{Examples}\\
&&&\\
\hline
\multirow{2}{*}{No physical constraints \cite{part2}}&\multirow{2}{*}{$\backslash$}&\multirow{2}{*}{$ \|\bW\|_F^2\leq P_T$}&can implement specific beamforming matrices, e.g.,\\
&&& zero-forcing, with a complexity of $\mathcal{O}(NK)$\\
\hline
\multirow{2}{*}{Passive}&\multirow{2}{*}{$\bthe^H\bthe\preceq \mathbf{I}_{N+K}$}&\multirow{2}{*}{$\bW=\mathbf{F}\text{diag}(\mathbf{p})^{{1}/{2}}, \|\bF\|_2\leq 1,~\mathbf{1}^T\mathbf{p}\leq P_T$} &column-orthogonal \\
&&&(semi-unitary) matrix\\
\hline
\multirow{2}{*}{Lossless and reciprocal} &\multirow{2}{*}{$\bthe^H\bthe=\mathbf{I}_{N+K},~\bthe=\bthe^T$}&\multirow{2}{*}{$\bW=\mathbf{F}\text{diag}(\mathbf{p})^{{1}/{2}}, \|\bF\|_2\leq 1,~\mathbf{1}^T\mathbf{p}\leq P_T$} &column-orthogonal \\
&&&(semi-unitary) matrix\\
\hline
\end{tabular}
\label{summary}}
\centering
\end{table*}

A natural question now arises: What kind of digital beamforming matrices cannot be achieved by MiLAC-aided beamforming? Proposition \ref{proposition2} (ii) provides one such example: MiLAC-aided beamforming cannot realize full-power beamforming matrices whose columns are not unitary. We offer further insights in the following remark. 
\begin{remark}\label{remark1}
To compare $\mathcal{W}_{\text{MiLAC}}$ with $\mathcal{W}_{\text{digital}}$, we first rewrite $\mathcal{W}_{\text{digital}}$ as 
\begin{equation}\label{digital:F}
\mathcal{W}_{\text{digital}}=\{\bW=\bF\text{\normalfont diag}(\mathbf{p})^{\frac{1}{2}}\mid \max_{i}\|\mathbf{f}_{i}\|_2\leq 1,  \mathbf{1}^T\mathbf{p}\leq P_T\}.
\end{equation}
The difference between $\mathcal{W}_{\text{MiLAC}}$ and $\mathcal{W}_{\text{digital}}$ lies in the constraint on $\bF$. In $\mathcal{W}_{\text{MiLAC}}$, the matrix $\bF$ is required to satisfy
\begin{equation}\label{milac:F}
\|\bF\|_2=\max_{\|\x\|_2=1}\left\|\sum_{i=1}^Kx_i\mathbf{f}_i\right\|\leq 1.
\end{equation}
Since $\|\bF\|_2\geq  \max_{i}\|\mathbf{f}_{i}\|_2$, the constraint in \eqref{milac:F} is more restrictive than that in \eqref{digital:F}. Moreover, this inequality is generally strict. For example, consider the matrix $\bF=[\mathbf{f},\mathbf{f},\dots,\mathbf{f}]$ with $\|\mathbf{f}\|_2=1$. It satisfies $\max_{i}\|\mathbf{f}_{i}\|_2= 1$, but $\|\bF\|_2=\sqrt{K}$ and thus  \eqref{milac:F} does not hold. 
This again shows that $\mathcal{W}_{\text{MiLAC}}\subsetneqq \mathcal{W}_{\text{digital}}$. 

Intuitively, the difference arises because the MiLAC constraint in \eqref{milac:F} accounts for the correlation among the columns of $\bF$, whereas \eqref{digital:F} only constrains the norm of each column individually.  Consequently, $\mathcal{W}_{\text{MiLAC}}$ excludes matrices in $\mathcal{W}_{\text{digital}}$ whose columns are highly correlated.
\end{remark}

Remark \ref{remark1} further implies that MiLAC-aided beamforming can achieve performance very close to digital beamforming when the users have nearly orthogonal channels, since in this case the optimal digital beamforming matrices tend to have nearly orthogonal columns and thus are well approximated within $\mathcal{W}_{\text{MiLAC}}$. This situation naturally arises in gigantic MIMO systems, where the large number of antennas typically makes user channels asymptotically orthogonal.

\subsection{Comparison with Phase-Shifter-Based Analog Beamforming}\label{sec:3c}
Traditional analog beamforming, including the analog component in hybrid beamforming architectures, is realized using phase shifters and can be modeled by a matrix $\bW_A\in\C^{N\times K}$ with a constant-modulus constraint, i.e., 
\begin{equation}\label{analog}
|[\bW_A]_{i,j}|=\frac{1}{\sqrt{NK}},~\forall~i,j.
\end{equation} Here, the normalization is introduced to guarantee that the phase-shifter network is power normalized and does not amplify the input signal. 

Alternatively, analog beamforming can be realized via MiLAC. From Proposition \ref{pro1} and noting that any matrix characterized by \eqref{analog} satisfies  $\|\bW_A\|_2\leq \|\bW_A\|_F=1$, we can conclude that the class of analog beamforming matrices achieved by phase shifters is strictly contained within the class achievable by MiLAC.  In other words, MiLAC offers greater flexibility for analog beamforming than conventional phase shifters. 

\subsection{A Digital-Achieving Hybrid Digital–MiLAC Architecture}\label{sec:3d}
The discussion in Section \ref{sec:3c} motivates a hybrid digital–MiLAC beamforming scheme, in which  the phase shifters used in conventional hybrid beamforming is replaced with MiLAC; see Fig. \ref{scheme2} (d). This scheme provides greater flexibility than traditional hybrid beamforming with the same number of RF chains. Moreover, we show in the following proposition that a hybrid digital-MiLAC beamforming scheme achieves the same flexibility as digital beamforming with $K$ RF chains.

\begin{proposition}
Any matrix $\bW\in\mathcal{W}_{\text{digital}}$ can be achieved by a hybrid digital-MiLAC architecture with $K$ RF chains. 
\end{proposition}
\begin{proof}
The hybrid digital-MiLAC beamforming matrix can be written as 
$\bW=\mathbf{F}\mathbf{P}$, where $\bF\in\C^{N\times K}$ satisfies the MiLAC-induced constraint $\|\bF\|_2\leq 1$ and $\bP\in\C^{K\times K}$ satisfies the transmit power constraint $\|\bP\|_F^2\leq P_T$. Given any digital beamforming matrix $\bW$ with $\|\bW\|_F^2\leq P_T$,  let $\bW=\bU_W\bD_W\bV_W^H$ be its SVD. Then $\bW$ can be realized by the hybrid digital–MiLAC architecture by setting $\bF=[\bU_W]_{:,1:K}$ and $\bP=[\bD_{W}]_{1:K,:}\bV_{W}^H$.

\end{proof}
It is well known that conventional hybrid beamforming requires $2K$ RF chains to achieve fully digital beamforming \cite{hybrid2}. The proposed hybrid digital–MiLAC architecture reduces the required number of RF chains to only $K$, thanks to the greater flexibility provided by MiLAC compared with phase shifters.


\section{MiLAC-aided Beamforming Design}\label{sec:4}
In the previous section, we characterized the MiLAC-aided beamforming matrices and showed that MiLAC-aided beamforming cannot achieve the full flexibility of digital beamforming. This section aims to design effective and efficient algorithms for MiLAC-aided beamforming.


\subsection{A Convex Reformulation of $\mathcal{W}_{\text{MiLAC}}$ }
We begin by  transforming $\mathcal{W}_{\text{MiLAC}}$ into a convex set. 
\begin{lemma}\label{lemma1}
$\mathcal{W}_{\text{MiLAC}}$ can be equivalently expressed as the following convex set:
\begin{equation}\label{MiLAC3}
\mathcal{W}_{\text{MiLAC}}=\{\bW\mid \bW^H\bW\preceq\text{diag}(\mathbf{p}),~\mathbf{1}^T\mathbf{p}\leq P_T\}.
\end{equation}
\end{lemma}
\begin{proof}
We use $\mathcal{W}_{\text{MiLAC}}^1$ and $\mathcal{W}_{\text{MiLAC}}^2$ to denote the set in \eqref{MiLAC2} and \eqref{MiLAC3}, respectively.  The goal is to show that $\mathcal{W}_{\text{MiLAC}}^1=\mathcal{W}_{\text{MiLAC}}^2$.

First, {$\mathcal{W}_{\text{MiLAC}}^1\subseteq\mathcal{W}_{\text{MiLAC}}^2$} follows from \eqref{proof:eq1}. We next prove that   \emph{$\mathcal{W}_{\text{MiLAC}}^2\subseteq\mathcal{W}_{\text{MiLAC}}^1$}. Given $\bW\in\mathcal{W}_{\text{MiLAC}}^2$, we define 
\begin{equation}\label{proof2:eq1}
\bF=\bW\text{diag}(\mathbf{p}^\dagger)^{\frac{1}{2}},
\end{equation}
where $p^\dagger_i=p_i^{-1}$ if $p_i\neq 0$, and $p^\dagger_i=0$ otherwise.  With \eqref{proof2:eq1}, we claim that 
$\bW=\bF\text{diag}(\mathbf{p})^{\frac{1}{2}}.$
Specifically, comparing each column of both sides of \eqref{proof2:eq1} and using the definition of $\mathbf{p}^{\dagger}$ yield
\begin{equation}\label{proof2:eq2}
\mathbf{f}_{i}=
\left\{ \begin{aligned}\frac{\bw_i}{\sqrt{p_i}},~&\text{if }p_i>0,~~\\
\mathbf{0},~~&\text{if }p_i=0.
\end{aligned}\right.
\end{equation}
 Hence, if $p_i>0$, $\bw_i=\mathbf{f}_i\sqrt{p_i}$. If $p_i=0$, it follows from $ \bW^H\bW\preceq\text{diag}(\mathbf{p})$ that $\bw_i=\mathbf{0}$; see \eqref{proof:eq2a}, and thus $\mathbf{w}_i=\mathbf{f}_i\sqrt{p_i}$, which proves the claim. 
 In addition, 
$$
\begin{aligned}
\bF^H\bF&=\text{diag}(\mathbf{p}^\dagger)^{1/2}\bW^H\bW\text{diag}(\mathbf{p}^\dagger)^{1/2}\\
&\preceq\text{diag}(\mathbf{p}^\dagger)^{1/2}\text{diag}(\mathbf{p})\text{diag}(\mathbf{p}^\dagger)^{1/2}\preceq\mathbf{I},
\end{aligned}$$
i.e., $\|\bF\|_2\leq 1$. Therefore, $\bW\in\mathcal{W}_{\text{MiLAC}}^1$.

\end{proof}
With Lemma \ref{lemma1}, the sum-rate maximization problem in \eqref{sumrate} can be equivalently expressed as 
\begin{subequations}\label{sumrate2}
\begin{align}
\max_{\bW,\mathbf{p}}~&\sum_{k=1}^K\log\left(1+\frac{|\h_k^H\mathbf{w}_{k}|^2}{\sum_{j\neq k}|\h_k^H\bw_{j}|^2+\sigma^2}\right)\\
\text{s.t. }~&\bW^H\bW\preceq\text{diag}(\mathbf{p}),\label{con:1}\\
&\mathbf{1}^T\mathbf{p}\leq P_T\label{con:2},
\end{align}
\end{subequations}
where $\sigma^2:=4\sigma_n^2$. This formulation resembles the classical sum-rate maximization problem with fully digital beamforming, except that the  power constraint $\|\bW\|_F^2\leq P_T$ is now replaced by \eqref{con:1} and \eqref{con:2}. 
Similar to the digital beamforming case, we can apply the weighted minimum mean square error (WMMSE) \cite{wmmse} technique to transform the objective function in to a block-convex form and then solve the problem using block coordinate descent (BCD) algorithm. The only difference lies in the $(\bW,\mathbf{p})$-subproblem, which now includes the new constraints in \eqref{con:1} and \eqref{con:2}. In particular, the constraint \eqref{con:1} can be transformed into a linear matrix inequality (LMI) as \cite{boyd2004convex}
 \begin{equation}\label{LMI}
\left[\begin{matrix} \mathbf{I}_{N}&\bW\\\bW^H&\text{diag}(\mathbf{p})\end{matrix}\right]\succeq \mathbf{0}.
\end{equation}
Hence, the resulting $(\bW,\mathbf{p})$-subproblem is a semidefinite program with a quadratic objective function and can be solved to global optimality by CVX. 

 However, solving a semidefinite program with dimension $(N+K)$ at each iteration is computationally expensive, particularly in the context of gigantic MIMO. 
 In the following subsection, we exploit a low-dimensional subspace property of problem \eqref{sumrate2} to reduce its dimensionality, thereby enabling more efficient optimization algorithms.
 
   \subsection{A Low-Dimensional Reformulation of \eqref{sumrate2}}\label{sec:4b}
  In this section, we exploit a nice low-dimensional subspace property of problem \eqref{sumrate2}.  The idea is motivated by a similar property established for digital beamforming in \cite{rethinkingwmmse,rethinkingwmmse2}. However, due to the different constraints introduced by MiLAC, our problem requires dedicated analysis.
  
   In the following lemma, we show that it suffices to search for the  solutions of problem \eqref{sumrate2} in the low-dimensional subspace spanned by $\text{Ran}(\bH^H)$, where $\bH^H=[\h_1,\h_2,\dots,\h_K]$. 
     
\begin{lemma}\label{pro4}
For any feasible point $(\bW,\mathbf{p})$ of \eqref{sumrate2}, 
$(\Pi_{\mathbf{H}^H}\mathbf{W},\mathbf{p})$ is also feasible for \eqref{sumrate2} and achieves the same objective value as $(\bW,\mathbf{p})$, where $\Pi_{\mathbf{H}^H}\bW=\bH^H(\bH\bH^H)^{-1}\bH\bW$ is the orthogonal projection of $\bW$ onto the range space of $\bH^H$.
\end{lemma}
\begin{proof}
Express $\bW$ as 
$
\mathbf{W}= \widehat{\mathbf{W}} + \widetilde{\mathbf{W}}
$
where 
$\widehat{\mathbf{W}}= \Pi_{\mathbf{H}^H}\mathbf{W}$  and $\widetilde{\mathbf{W}}= \bigl(\mathbf{I} - \Pi_{\mathbf{H}^H}\bigr)\mathbf{W}$ are the orthogonal projection of $\bW$ onto the range space of $\bH^H$  and null space of $\bH$, respectively. With this representation, the objective function achieved by $\bW$ is 
\begin{equation}\label{obj}
\sum_{k=1}^K\log\left(1+\frac{|\h_k^H\hat{\mathbf{w}}_{k}|^2}{\sum_{j\neq k}|\h_k^H\hat{\bw}_{j}|^2+\sigma^2}\right),
\end{equation} where we have used the fact that $\bH\widetilde{\bW}=\mathbf{0}$. %
Since $(\bW,\mathbf{p})$ is feasible, we obtain
\begin{equation}\label{constraint}
\widehat{\mathbf{W}}^H\widehat{\mathbf{W}} \preceq \text{diag}(\mathbf{p})-\widetilde{\mathbf{W}}^H\widetilde{\mathbf{W}}\preceq\text{diag}(\mathbf{p}).
\end{equation}
Therefore, $(\widehat{\mathbf{W}},\mathbf{p})$ achieves the same objective value as $(\bW,\mathbf{p})$, given by \eqref{obj}, and  satisfies the constraints of  problem \eqref{sumrate2} due to \eqref{constraint}, which proves the lemma.
\end{proof}
As implied by Lemma \ref{pro4}, there exists an optimal solution $(\bW^*,\mathbf{p}^*)$ to \eqref{sumrate2} with $\bW^*\in\text{Ran}(\bH^H)$. Hence, the optimization problem can be solved without loss of optimality by restricting $\bW$ to $\text{Ran}(\bH^H)$. Note that any $\bW\in\text{Ran} (\bH^H)$ can be expressed as $\bW=\bH^H\bX$, where $\bX\in\C^{K\times K}$. Let $\overline{\bH}=\bH\bH^H$, we formulate the following low-dimensional problem:
\begin{subequations}\label{sumrate3}
\begin{align}
\max_{\bX,\mathbf{p}}~&\sum_{k=1}^K\log\left(1+\frac{|\bar{\h}_k^H\mathbf{x}_{k}|^2}{\sum_{j\neq k}|\bar{\bh}_k^H\x_{j}|^2+\sigma^2}\right)\\
\text{s.t. }~&\bX^H\overline{\bH}\bX\preceq\text{diag}(\mathbf{p}),~\label{con:X}\\
&\mathbf{1}^T\mathbf{p}\leq P_T,\label{con:p}
\end{align}
\end{subequations}
where $\bar{\bh}_k^H$ and $\mathbf{x}_k$ are the $k$-th row and $k$-th column of $\bar{\bH}$ and $\bX$, respectively. 
Compared to the original formulation in \eqref{sumrate2}, the variable $\bW\in\C^{N\times K}$ is replaced by the new variable $\bX\in\C^{K\times K}$ in \eqref{sumrate3}. This substantially reduces the dimension of the optimization variables,  especially in gigantic MIMO systems where $N\gg K$.

We remark that if $\bH$ is rank-deficient, the problem dimension can be further reduced by expressing  $\bW=\bA\bX$, where $\bA\in\C^{N\times \text{rank}(\bH)}$ spans $\text{Ran}(\bH^H)$ and $\bX\in\C^{\text{rank}(\bH)\times K}$. Throughout the remainder of the paper, we assume without loss of generality that $\bH$ has full row rank.

Since the objective functions of \eqref{sumrate2} and \eqref{sumrate3} are non-convex, optimization algorithms can typically guarantee convergence only to a stationary point. The following proposition further establishes the relationship between the stationary points of the two problems.\begin{proposition}\label{sp_eq}
Let $(\bW,\mathbf{p})$ be a stationary point of problem \eqref{sumrate2}.  Then $(\bar{\bH}^{-1}\bH\bW,\mathbf{p})$ is a stationary point of problem \eqref{sumrate3}.   On the other hand,  let $(\bX,\mathbf{p})$ be a stationary point of \eqref{sumrate3}, then $(\bH^H\bX,\mathbf{p})$ is a stationary point of \eqref{sumrate2}.
\end{proposition}
\begin{proof}
See Appendix \ref{app:proof:sp_eq}.
\end{proof}
Building on Proposition \ref{sp_eq}, we now turn to solving problem \eqref{sumrate3} in the following subsections. 
\subsection{WMMSE Approach for Solving \eqref{sumrate3}}
In this subsection, we present the classical WMMSE approach for solving \eqref{sumrate3}.  
By applying the same transformation  as in  \cite{wmmse},  problem \eqref{sumrate3} can be equivalently reformulated in the following weighted mean square error (MSE)  minimization form:
\begin{equation}\label{problem:wmmse}
\begin{aligned}
\min_{\boldsymbol{\omega},\boldsymbol{u},\bX,\mathbf{p}}~&\sum_{k=1}^K\left(\omega_kE_k(u_k,\bX)-\log(\omega_k)\right)\\
\text{s.t.}~~~~&\eqref{con:X}\text{ and }\eqref{con:p},
\end{aligned}
\end{equation}
where 
$$E_k(u_k,\bX):=|1-u_k^*\bar{\h}_k^H\x_k|^2+|u_k|^2\left(\sum_{j\neq k}|\bar{\h}_k^H\x_j|^2+\sigma^2\right)$$
is the MSE between symbol $s_k$ and the estimated symbol $\hat{s}_k:=u_k^*y_k$. We can now apply the BCD algorithm to iteratively minimize the objective function of \eqref{problem:wmmse}. The updates of $\boldsymbol{\omega}$ and $\u$ are standard and are exactly the same as in the digital beamforming case: 
\begin{equation}\label{uk}
u_k=\frac{\bar{\h}_k^H\x_k}{\sum_{j=1}^K|\bar{\h}_k^H\x_j|^2+\sigma^2},~k=1,2,\dots, K,
\end{equation}
and 
\begin{equation}\label{omegak}
\omega_k=(1-u_k^{*}\bar{\h}_k^H\x_k)^{-1},~k=1,2,\dots, K,
\end{equation}
which are obtained by minimizing the objective function in \eqref{problem:wmmse} with respective to $\boldsymbol{\omega}$ and ${\u}$, respectively, with the other variable blocks fixed. While fixing $\boldsymbol{\omega}$ and $\mathbf{u}$, $(\bX,\mathbf{p})$ are updated by solving the following problem: 
\begin{equation}\label{X-subproblem}
\begin{aligned}
\min_{\bX,\mathbf{p}}~& \text{tr}(\bX^H\bQ\bX)-2\RR\left(\text{tr}(\mathbf{L}\bX)\right)\\
\text{s.t.}~~&\eqref{con:X}\text{ and }\eqref{con:p},
\end{aligned}
\end{equation}
where $\bQ=\text{diag}({\boldsymbol{\omega}}\circ\u\circ\u^*)\bar{\bH}\bar{\bH}^H$ and $\mathbf{L}=\text{diag}(\boldsymbol{\omega}\circ\u^*)\bar{\bH}$. Similar to \eqref{LMI}, constraint \eqref{con:X} can be rewritten as a LMI: 
 \begin{equation}\label{LMI2}
\left[\begin{matrix} \mathbf{I}_{K}&\bar{\bH}^{\frac{1}{2}}\bX\\\bX^H\bar{\bH}^{\frac{1}{2}}&\text{diag}(\mathbf{p})\end{matrix}\right]\succeq \mathbf{0},
\end{equation}
where $\bar{\bH}^{\frac{1}{2}}$ is the square root of $\bar{\bH}$.
With \eqref{LMI2}, problem \eqref{X-subproblem} transforms to a SDP and can be solved  by CVX. We summarize the overall algorithm in Algorithm \ref{WMMSE}. The LMI in \eqref{LMI2} has a dimension of $2K$, which is much lower than the $(N+K)$-dimensional LMI in \eqref{LMI} that appears in the original formulation involving the variable $\bW$. This significantly improves the efficiency of the WMMSE algorithm.  In the following, we discuss the complexity and convergence of Algorithm \ref{WMMSE}.

\emph{Complexity:} First, constructing $\bar{\bH}$ has a complexity of $\mathcal{O}(NK^2)$. In Algorithm \ref{WMMSE}, the complexity of updating $\mathbf{u}$ and $\boldsymbol{\omega}$ is dominated by calculating $\bar{\bH}^H\bX$, which has a complexity of   $\mathcal{O}(K^3)$. For the $\bX$ subproblem, solving the SDP incurs a per-iteration complexity of $\mathcal{O}(K^6)$.  Let $I_{\text{out}}$ be the total number of iterations for Algorithm \ref{WMMSE}, and let $I_{\text{inner}}^{(i)}$ be the number of iterations for solving the  $\bX$-subproblem at the $i$-th outer iteration. The total complexity of Algorithm \ref{WMMSE} is therefore $\mathcal{O}(NK^2+K^6\sum_{i=1}^{I_{\text{out}}}I_{\text{inner}}^{(i)}).$

In contrast, the WMMSE algorithm applied directly to the original formulation \eqref{sumrate2} requires solving an SDP of dimension $N+K$ at each iteration, resulting in a total complexity of $\mathcal{O}((N+K)^6\sum_{i=1}^{I_{\text{out}}}I_{\text{inner}}^{(i)}),$ which is significantly higher.

\emph{Convergence:} We establish the convergence of Algorithm \ref{WMMSE} in the following theorem.  
\begin{theorem}\label{converge1}
Any limit point $(\mathbf{u}^*, \boldsymbol{\omega}^*,\bX^*,\mathbf{p}^*)$ of the sequence generated by Algorithm \ref{WMMSE} is a stationary point of problem \eqref{problem:wmmse}. Further, $(\bX^*,\mathbf{p}^*)$ is a stationary point of problem \eqref{sumrate3}.  
\end{theorem}
\begin{proof}
See Appendix \ref{app:convergence}.
\end{proof}
Although the computational complexity of WMMSE is greatly reduced by exploiting the low-dimensional structure, it remains computationally expensive in practice due to the requirement to solve an SDP at each iteration, even for small values of $K$. In the next subsection, we propose a low-complexity algorithm for solving \eqref{sumrate3} that  avoids solving any SDP.
\begin{algorithm}[t]
\caption{The WMMSE Algorithm for Solving \eqref{sumrate3}}
\label{alg:wmmse}
\begin{algorithmic}[1]
\STATE \textbf{Initialization:} Initialize $\bX=\sqrt{\frac{P_T}{K\|\bH\|^2_2}}\mathbf{I}_K$.
\STATE Set the tolerance of accuracy $\epsilon > 0$.
\REPEAT
\STATE $\omega_k'=\omega_k$,~$k=1,2,\dots, K$;
    \STATE $u_k=\frac{\bar{\h}_k^H\x_j}{\sum_{j=1}^K|\bar{\h}_k^H\x_k|^2+\sigma^2},~k=1,2,\dots, K$;
    \STATE $\omega_k=(1-u_k^{*}\bar{\h}_k^H\x_k)^{-1},~k=1,2,\dots, K$;

    \STATE Update $\bX$ and $\mathbf{p}$ by solving problem \eqref{X-subproblem};
\UNTIL{$\frac{\left|\sum_{k=1}^K\log(\omega_k)
      - \sum_{k=1}^K  \log(\omega'_k)\right|}{\sum_{k=1}^K\log(\omega'_k)} \le \epsilon$}
\STATE \textbf{Output:} $\bX$ and $\mathbf{p}$.
\end{algorithmic}
\label{WMMSE}
\end{algorithm}
\subsection{A Low-Complexity Algorithm for Solving \eqref{sumrate3}}\label{sec:4c}
The idea is to reformulate the complex LMI constraint in \eqref{con:X} into separate simpler constraints on $\bX$ and $\mathbf{p}$, and then update $\bX$ and $\mathbf{p}$ independently within the BCD framework.   Similar  to the proof of Lemma \ref{lemma1}, we can show that 
$$
\begin{aligned}
&\{\bX\mid \bX^H\bar{\bH}\bX\preceq\text{diag}(\mathbf{p}),~\mathbf{1}^T\mathbf{p}\leq P_T\}\\
&=\{\bX=\bar{\bH}^{-\frac{1}{2}}\mathbf{Y}\text{diag}(\mathbf{p})^{\frac{1}{2}}\mid \|\bY\|_2\leq 1, ~\mathbf{1}^T\mathbf{p}\leq P_T\}. 
\end{aligned}
$$ 
Using this reformulation and  denoting $\widehat{\bH}=\bar{\bH}^{\frac{1}{2}}$,  problem \eqref{sumrate3} transforms to
\begin{subequations}\label{sumrate4}
\begin{align}
\max_{\mathbf{F},\mathbf{p}}~&\sum_{k=1}^K\log_2\left(1+\frac{|\hat{\h}_k^H{\mathbf{y}}_{k}|^2p_k}{\sum_{j\neq k}|\hat{\bh}_k^H\mathbf{y}_{j}|^2p_j+\sigma^2}\right)\\
\text{s.t. }~&\|\mathbf{Y}\|_2\leq 1,~\label{con:Y}\\
&\mathbf{1}^T\mathbf{p}\leq P_T,\label{con:p2}
\end{align}
\end{subequations}
where $\hat{\h}_k^H$ and $\mathbf{y}_k$ are the $k$-th row and $k$-th column of $\widehat{\bH}$ and $\bY$, respectively.  
Applying the transformation in \eqref{problem:wmmse} gives
\begin{equation}\label{problem:wmmse2}
\begin{aligned}
\min_{\boldsymbol{\omega},\boldsymbol{u},\bY,\mathbf{p}}~&\sum_{k=1}^K\left(\omega_kE_k(u_k,\mathbf{p},\bY)-\log(\omega_k)\right)\\
\text{s.t.}~~~~&\text{\eqref{con:Y} and \eqref{con:p2}}
\end{aligned}
\end{equation}
where 
$$
\begin{aligned}
E_k(u_k,\mathbf{p},\bY)=&|1-\sqrt{p_k}u_k^*\hat{\h}_k^H\y_k|^2\\
&+|u_k|^2\left(\sum_{j\neq k}|\hat{\h}_k^H\y_j|^2p_j+\sigma^2\right)
\end{aligned}$$
We still use BCD algorithm to solve \eqref{problem:wmmse2}, but update $\bY$ and $\mathbf{p}$ separately. 
Similar to   \eqref{uk} and \eqref{omegak}, variables $\u$ and $\boldsymbol{\omega}$  are updated as 
$u_k=\frac{\hat{\h}_k^H\y_k\sqrt{p_k}}{\sum_{j=1}^K|\hat{\h}_k^H\x_j|^2p_j+\sigma^2},~k=1,2,\dots, K,
$and 
\begin{equation*}
\omega_k=(1-u_k^{*}\hat{\h}_k^H\y_k\sqrt{p_k})^{-1},~k=1,2,\dots, K,
\end{equation*}
 respectively. When $\u$, $\boldsymbol{\omega}$, and $\mathbf{Y}$ are fixed, the $\mathbf{p}$-subproblem is 
\begin{equation}\label{p-subproblem}
\begin{aligned}
\max_{\mathbf{p}}~&\sum_{k=1}^K(2\alpha_k\sqrt{p_k}-\beta_kp_k),~\\
\text{s.t. }~&\mathbf{1}^T\mathbf{p}\leq P_T.\end{aligned}
\end{equation}
where $\alpha_k=\omega_k\RR(u_k^*\hat{\h}_k^H\y_k)$ and $\beta_k=\sum_{j=1}^K\omega_j|u_j|^2|\hat{\h}_j^H\mathbf{y}_k|^2$.
The solution to \eqref{p-subproblem} can be derived from its Karush–Kuhn–Tucker (KKT) conditions. Let $\lambda$ be the Lagrange multiplier associated with $\mathbf{1}^T\mathbf{p}\leq P_T$. The KKT condition gives  
\begin{equation}\label{update:pk}
\begin{aligned}
p_k=\frac{\alpha_k^2}{(\beta_k+\lambda)^2},~~k=1,2,\dots,K,
\end{aligned}
\end{equation}
where $\lambda$ and $\mathbf{p}$ satisfy 
\begin{subequations}\label{kkt:pk}
\begin{align}
&\lambda(\mathbf{1}^T\mathbf{p}-P_T)=0, \label{complementary}\\
&\mathbf{1}^T\mathbf{p}\leq P_T,~{\lambda}\geq {0},\label{feasibility}
\end{align}
\end{subequations} where \eqref{complementary} is the complementary slackness condition and \eqref{feasibility} is the feasibility condition in the primal and dual domain. Combining \eqref{update:pk} and \eqref{kkt:pk}, we can conclude that $\lambda=0$ if $\sum_{k=1}^K\alpha_k^2/\beta_k\leq P_T$. Otherwise, $\lambda$ is the unique positive solution to
 $$\sum_{k=1}^K\frac{\alpha_k^2}{(\beta_k+\lambda)^2}=P_T.$$

Finally, the variable $\bY$ is updated by solving the following problem with fixed $\u$, $\boldsymbol{\omega},$ and  $\mathbf{p}$: 
\begin{equation}\label{subproblemY}
\begin{aligned}
\min_{\bY}~&\text{tr}(\text{diag}(\mathbf{p})\mathbf{Y}^H\bQ\mathbf{Y})-2\RR\left(\text{tr}(\text{diag}(\mathbf{p})^{\frac{1}{2}}\mathbf{L}\mathbf{Y})\right)\\
\text{s.t.}~~&\|\mathbf{Y}\|_2\leq 1,
\end{aligned}
\end{equation}
where $\bQ=\text{diag}({\boldsymbol{\omega}}\circ\u\circ\u^*)\widehat{\bH}\widehat{\bH}^H$ and $\mathbf{L}=\text{diag}(\boldsymbol{\omega}\circ\u^*)\widehat{\bH}$ with $\widehat{\bH}=[\hat{\h}_1,\hat{\h}_2,\dots,\hat{\h}_K]^H$. Problem \eqref{subproblemY} has a quadratic objective function with a very simple constraint $\|\bY\|_2\leq 1$. To solve this subproblem efficiently, we adopt a projected gradient descent (PGD) method. Denote $\mathcal{Y}:=\{\bY\in\C^{K\times K}\mid\|\bY\|_2\leq 1\}$, which is the spectrum-norm ball that captures the feasible set of \eqref{subproblemY}. 
Following the update rule of PGD,  $\bY$ is updated as 
$$\bY_{t}=\text{Proj}_{\mathcal{Y}}\left(\bY_{t-1}-\eta\mathbf{G}({\bY_{t-1}})\right)$$ at the $t$-th iteration, 
where $\mathbf{G}({\bY}_{t-1})$ is the gradient of the objective function in \eqref{subproblemY} with respect to $\bY$ at point $\bY_{t-1}$, given by 
\begin{equation}\label{GY}
\mathbf{G}(\bY_{t-1})=\bQ\bY_{t-1}\text{diag}(\mathbf{p})-\mathbf{L}^H\text{diag}(\mathbf{p})^{\frac{1}{2}},
\end{equation}
$\eta>0$ is the stepsize, and 
$\text{Proj}_{\mathcal{Y}}(\bar{\bY})$ denotes the projection of $\bar{\bY}\in\C^{K\times K}$ onto set $\mathcal{Y}$ and is obtained by clipping the singular values of $\bar{\bY}$ to $1$ \cite{firstorder}.  Specifically, let $\bar{\bY}=\bU\bD\bV^H$ be the SVD of $\bY$, then
 $$\text{Proj}_{\mathcal{Y}}(\bar{\bY})=\bU\min\{\bD,1\}\bV^H,$$
 where $\min\{\bD,1\}$ denotes the element-wise minimum between $\bD$ and $1$, i.e., each singular value is clipped to at most $1$.
  For a convex problem like \eqref{subproblemY}, the PGD algorithm is guaranteed to converge to its optimal solution if the stepsize $\eta\leq L_G^{-1}$ \cite{}, where $L_G$ is a Lipschitz constant of the gradient $\mathbf{G}(\bY)$. The gradient $\bG(\bY)$ in \eqref{GY} satisfies 
  $$
  \begin{aligned}
  \|\bG(\bY_1)-\bG(\bY_2)\|_F&=\|\mathbf{Q}(\bY_1-\bY_2)\text{diag}(\mathbf{p})\|_F\\
 & \leq \|\bQ\|_2\|\text{diag}(\mathbf{p})\|_2\|\bY_1-\bY_2\|_F\\
 &= \|\bQ\|_2 \max_{k}\{p_k\}\|\bY_1-\bY_2\|_F.
  \end{aligned}
  $$
  Hence, $L_G=\|\bQ\|_2 \max_{k}\{p_k\}$ is a Lipschitz constant of $\bG(\bY)$, and we set $\eta=(\|\bQ\|_2\max_k\{p_k\})^{-1}$ in our algorithm.
  
The above procedure is summarized in Algorithm \ref{WMMSE-LC} and is referred to as the WMMSE-low-complexity (WMMSE-LC) algorithm. We next discuss its complexity and convergence properties. 

\emph{Complexity}: Computing $\widehat{\bH}$ requires a complexity of $\mathcal{O}(NK^2+K^3)$, where $\mathcal{O}(NK^2)$ accounts for computing $\bar{\bH}=\bH{\bH}^H$ and $\mathcal{O}(K^3)$ accounts for computing the square root of $\bar{\bH}$. Similar to Algorithm \ref{WMMSE}, the complexity of updating $\mathbf{u}$, $\boldsymbol{\omega}$, and $\mathbf{p}$ is dominated by calculating $\widehat{\bH}^H\bY$, which requires  $\mathcal{O}(K^3)$ operations.  For the $\bY$-subproblem, the initialization of computing $\bQ$, $\mathbf{L}$ and $\eta$ requires a complexity of $\mathcal{O}(K^3)$, and the per-iteration complexity of PGD algorithm is also $\mathcal{O}(K^3)$. Combining the above, the total complexity of Algorithm \ref{WMMSE-LC} is  $\mathcal{O}(NK^2+K^3\sum_{i=1}^{I_{\text{out}}}I_{\text{inner}}^{(i)}),$ where $I_{\text{out}}$ and $I_{\text{inner}}^{(i)}$ are the number of outer iterations for Algorithm \ref{WMMSE-LC} and the number of inner iterations for solving the $\bY$-subproblem at the $i$-th outer iteration, respectively.  

Compared with Algorithm \ref{WMMSE}, Algorithm \ref{WMMSE-LC} reduces the per-iteration complexity from 
$\mathcal{O}(K^6)$ to $\mathcal{O}(K^3)$. Moreover, the acceleration strategy in Algorithm \ref{WMMSE-LC} that splits the original convex constraint and alternately updates the two coupled variables  $\bY$ and $\mathbf{p}$ does not degrade the performance.  As will be demonstrated in the simulations, this low-complexity algorithm achieves nearly the same performance as Algorithm \ref{WMMSE}, while requiring significantly less CPU time.

\emph{Convergence:} We can get similar convergence results as Theorem \ref{converge1} for Algorithm \ref{WMMSE-LC}, as discussed in the following theorem.  
\begin{theorem}\label{converge2}
Any limit point $(\mathbf{u}^*, \boldsymbol{\omega}^*,\bY^*,\mathbf{p}^*)$ of the sequence generated by Algorithm \ref{WMMSE} is a stationary point of problem \eqref{problem:wmmse2}. Further, $(\bY^*,\mathbf{p}^*)$ is a stationary point of problem \eqref{sumrate4}.  
\end{theorem}
\begin{proof}
See Appendix \ref{app:convergence}.
\end{proof}

\begin{algorithm}[t]
\caption{The WMMSE-LC Algorithm for Solving \eqref{sumrate4}}
\label{alg:wmmse}
\begin{algorithmic}[1]
\STATE \textbf{Initialization:} Initialize $\mathbf{p}=\frac{P_T}{K}\mathbf{1}_K$ and $\bY=\widehat{\bH}/\|\hat\bH\|_2$.
\STATE Set the tolerance of accuracy for the out-loop and inner-loop as $\epsilon_{\text{out}} > 0$ and $\epsilon_{\text{in}}>0$.
\REPEAT
\STATE $\omega_k'=\omega_k$,~$k=1,2,\dots, K$;
    \STATE $u_k=\frac{\hat{\h}_k^H\y_k\sqrt{p_k}}{\sum_{j=1}^K|\hat{\h}_k^H\x_j|^2p_j+\sigma^2},~k=1,2,\dots, K$;
    \STATE $\omega_k=(1-u_k^{*}\hat{\h}_k^H\y_k\sqrt{p_k})^{-1},~k=1,2,\dots, K$;
        \STATE $p_k=\frac{\alpha_k^2}{(\beta_k+\lambda)^2},$ where $\alpha_k=\omega_k\RR(u_k^*\hat{\h}_k^H\y_k)$, $\beta_k=\sum_{j=1}^K\omega_j|u_j|^2|\hat{\h}_j^H\mathbf{y}_k|^2$,~ $k=1,2,\dots, K$, and $\lambda=0$ if $\sum_{k=1}^K\alpha_k^2/\beta_k\leq P_T$ and is the unique positive solution to
 $\sum_{k=1}^K\frac{\alpha_k^2}{(\beta_k+\lambda)^2}=P_T$ otherwise;

    \STATE Compute $\bQ=\text{diag}({\boldsymbol{\omega}}\circ\u\circ\u^*)\widehat{\bH}\widehat{\bH}^H$, $\mathbf{L}=\text{diag}(\boldsymbol{\omega}\circ\u^*)\widehat{\bH}$, and $\eta=\|\bQ\|_2\max_k\{p_k\}$;
    \REPEAT
\STATE $\bY'=\bY$;
    \STATE $\mathbf{G}(\bY)=\bQ\bY\text{diag}(\mathbf{p})-\mathbf{L}^H\text{diag}(\mathbf{p})^{\frac{1}{2}};$
    \STATE $\bar{\bY}=\bY-\eta\mathbf{G}({\bY})$;
    \STATE Compute the SVD of $\bar{\bY}$: $\bar{\bY}=\bU\bD\bV^H$;
    \STATE $\bY=\bU\min\{\bD,1\}\bV^H$.
        \UNTIL{$\|\bY-\bY'\|_F/\|\bY'\|_F \le \epsilon_{\text{in}}$}
\UNTIL{$\frac{\left|\sum_{k=1}^K\log(\omega_k)
      - \sum_{k=1}^K  \log(\omega'_k)\right|}{\sum_{k=1}^K  \log(\omega'_k) }\le \epsilon_{\text{out}}$}
\STATE \textbf{Output:} $\bY$ and $\mathbf{p}$.
\end{algorithmic}
\label{WMMSE-LC}
\end{algorithm}
\section{Simulation Results}\label{sec:simulation}
In this section, we present simulation results to evaluate the performance of the proposed algorithms in Section \ref{sec:4} and to compare the sum-rate of MiLAC-aided beamforming with that of digital and hybrid beamforming. 
\subsection{Comparison of the Proposed Algorithms}
We first compare the algorithms proposed in Section \ref{sec:4}. 
To validate the subspace property discussed in Section \ref{sec:4b}, we apply the WMMSE algorithm to both the original formulation \eqref{sumrate2} and the low-dimensional formulation \eqref{sumrate3}. The related results are labeled as ``WMMSE (full dim)” and ``WMMSE (reduced dim)” in the legend, respectively. We also include the low-complexity algorithm in Algorithm \ref{WMMSE-LC}, which is labeled as ``WMMSE-LC”.

\begin{figure}
\includegraphics[width=0.35\textwidth]{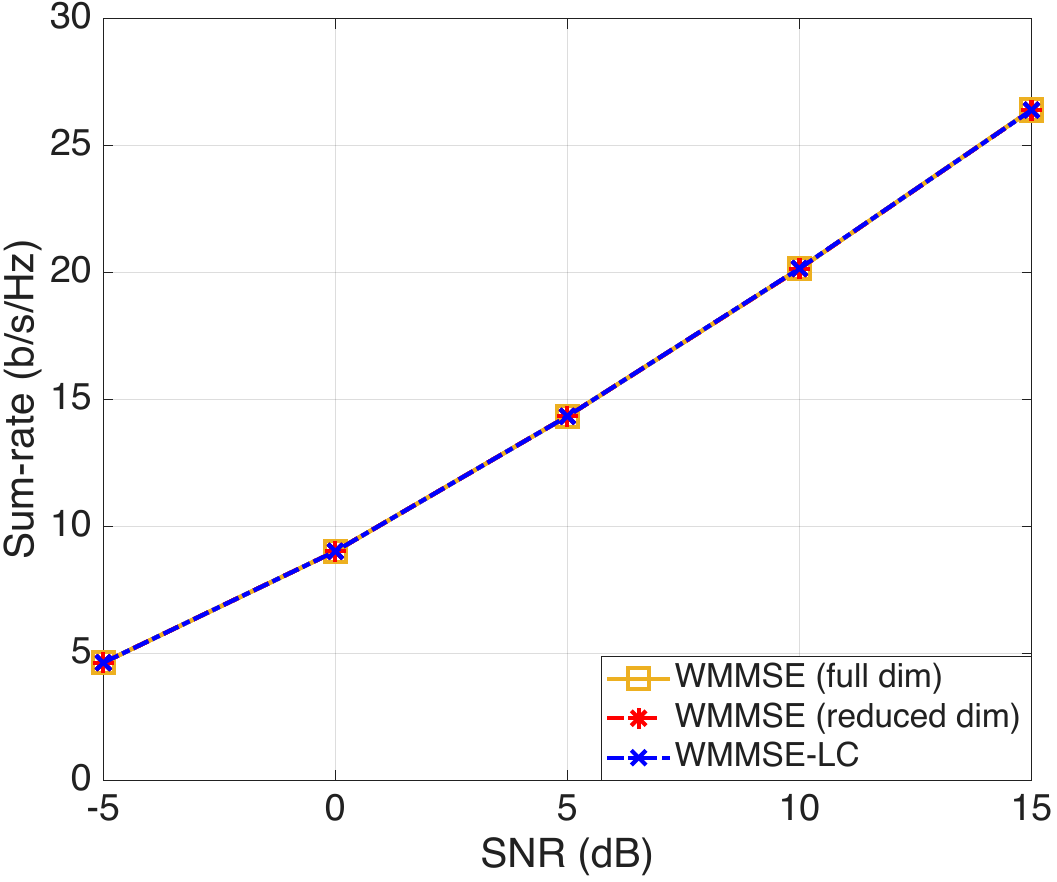}
\centering
\caption{Sum-rate performance for different algorithms, where $N=64$, $K=4$.}
\label{algorithm:sumrate}
\end{figure}
\begin{figure}
\subfigure[CPU time versus $N$, where $K=4$.]{\includegraphics[width=0.48\textwidth]{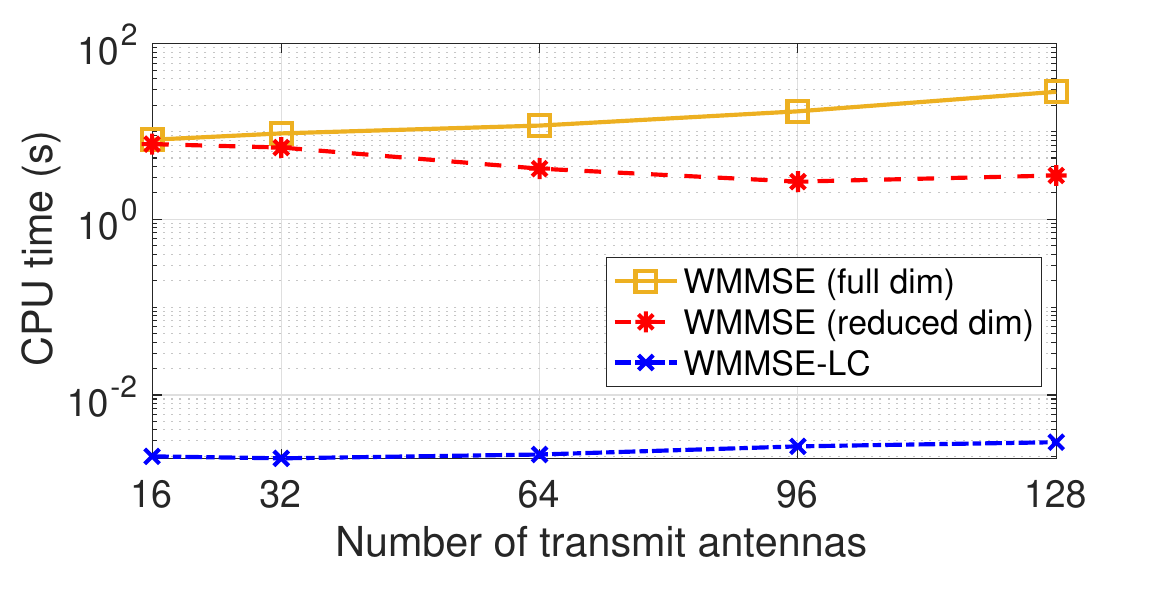}}
\subfigure[CPU time versus $K$, where $N=64$.]{\includegraphics[width=0.48\textwidth]{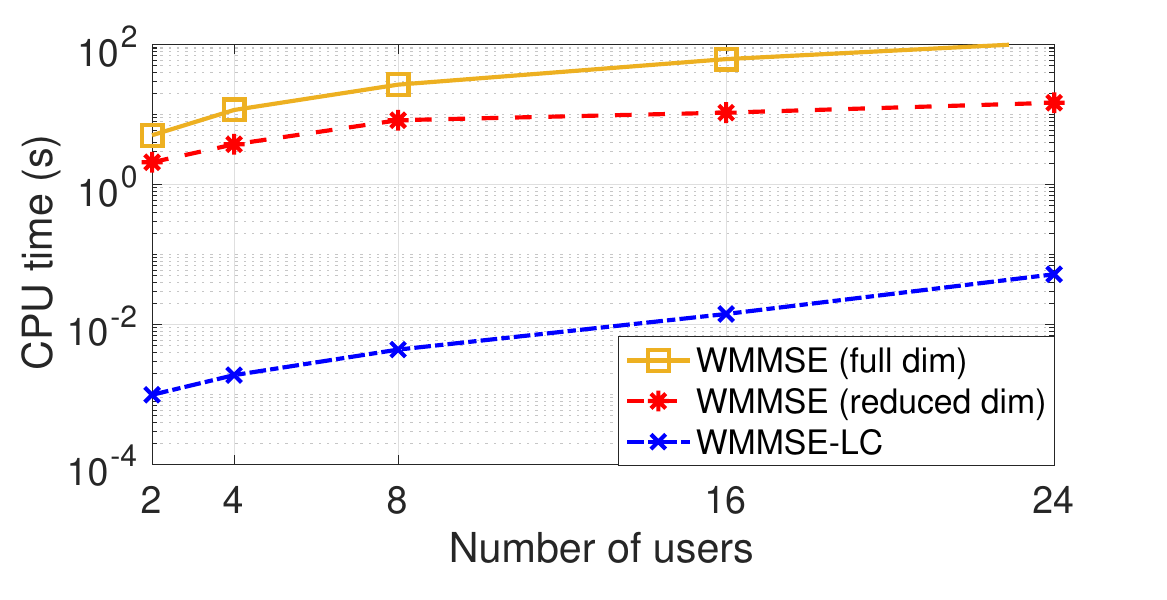}}
\centering
\caption{CPU time for different algorithms, where SNR$=15$\,dB.}
\label{algorithm:CPUtime}
\end{figure}
In Figs. \ref{algorithm:sumrate} and \ref{algorithm:CPUtime}, we compare the sum-rate and CPU time of different algorithms, where the channels are generated as Rayleigh fading, i.e., $\text{vec}(\bH)\sim\mathcal{CN}(\mathbf{0}_{NK\times 1},\mathbf{I}_{NK})$. As shown in the figure, the three algorithms achieve approximately the same performance, which validates the low-dimensional subspace property discussed in Section \ref{sec:4b} and the effectiveness of the low-complexity algorithm in Section \ref{sec:4c}. In terms of computational efficiency, however, these three approaches exhibit significantly different behaviors. As expected, the WMMSE applied to the low-dimensional formulation \eqref{sumrate3} is more efficient than when applied to the original formulation \eqref{sumrate2}, saving around 90\% CPU time when $N$ and $K$ are large. Interestingly, the CPU time of ``WMMSE (reduced dim)'' does not grow with $N$; it even shows a  slight decrease. This is because the computation complexity of the WMMSE algorithm for solving \eqref{sumrate3} is dominated by solving the $(\bW,\mathbf{p})$ subproblem, whose  dimension is irrelevant to $N$. In addition, the algorithm tends to converge faster as $N$ grows, since the problem becomes easier to solve when $\frac{N}{K}$ is large, thereby further reducing  the total CPU time. Compared with these two algorithms, the low-complexity implementation is significantly more efficient because it avoids solving any SDP, making it more favorable in the context of gigantic MIMO.

\subsection{Comparison of Different Beamforming Schemes}

In this subsection, we compare the sum-rate performance of digital, MiLAC-aided, and hybrid beamforming in MU-MISO systems.  
We consider both i.i.d. Rayleigh fading and clustered geometric  channels, in order to evaluate the three beamforming schemes under different propagation conditions. For the clustered geometric channel,  we adopt the model in \cite{hybrid2}:
$$\mathbf{h}_k 
= \sqrt{\frac{N}{L}}
\sum_{\ell=1}^{L} 
\alpha_k^\ell\, 
\mathbf{a}(\phi_{k}^\ell),~~k=1,2,\dots, K,
$$
where $\alpha_k^\ell \sim \mathcal{CN}(0,1)$ is the complex gain of the $\ell$-th path between the BS and user $k$ and $\phi_{k}^\ell \in [0,2\pi)$ is the angle of departure of user $k$. The array response vector is defined as  $
\mathbf{a}(\phi)
= \frac{1}{\sqrt{N}}
\left[
1,\;
e^{j \pi\sin\phi},\;
\dots,\;
e^{j \pi (N-1) \sin\phi}
\right]^{T},
$ which corresponds to a uniform linear array with half-wavelength antenna spacing. In our simulation, we set $L=5$, and randomly generated $\phi_k\in[0,2\pi),~\forall~k$.
\begin{figure}
\subfigure[Rayleigh fading channel.]{\includegraphics[width=0.35\textwidth]{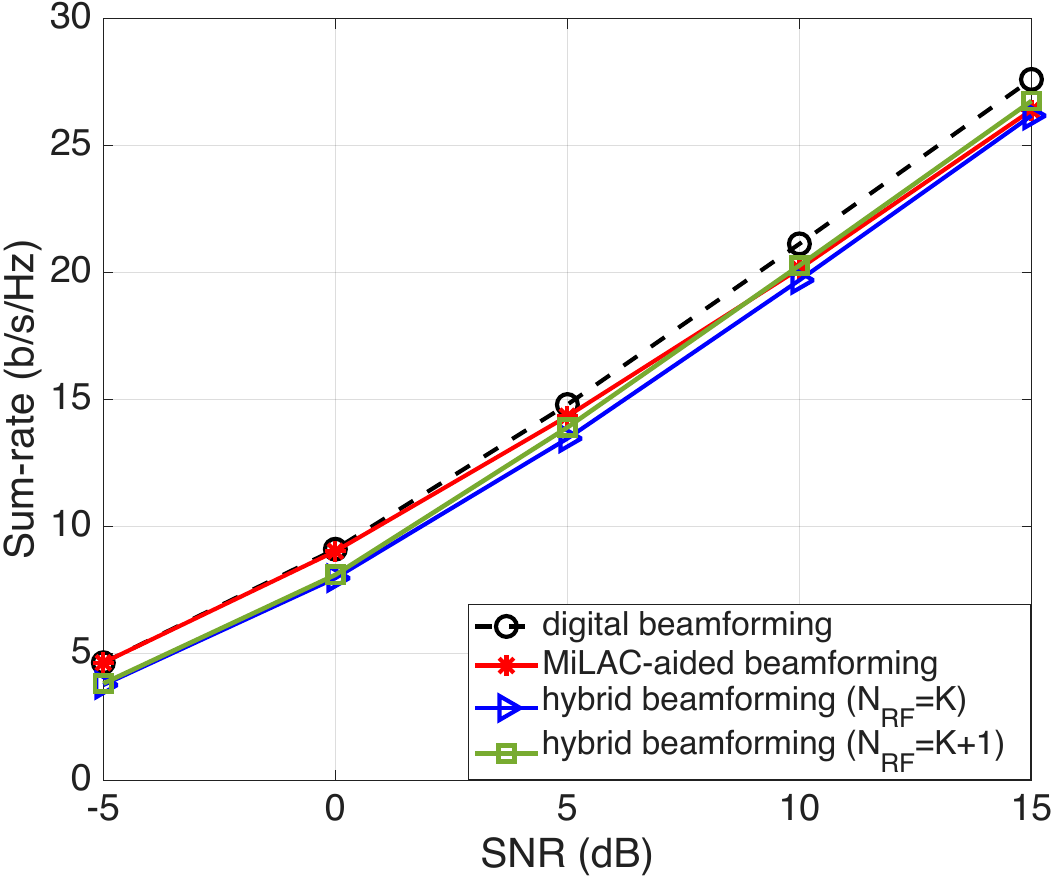}}
\subfigure[Clustered geometric channel.]{\includegraphics[width=0.35\textwidth]{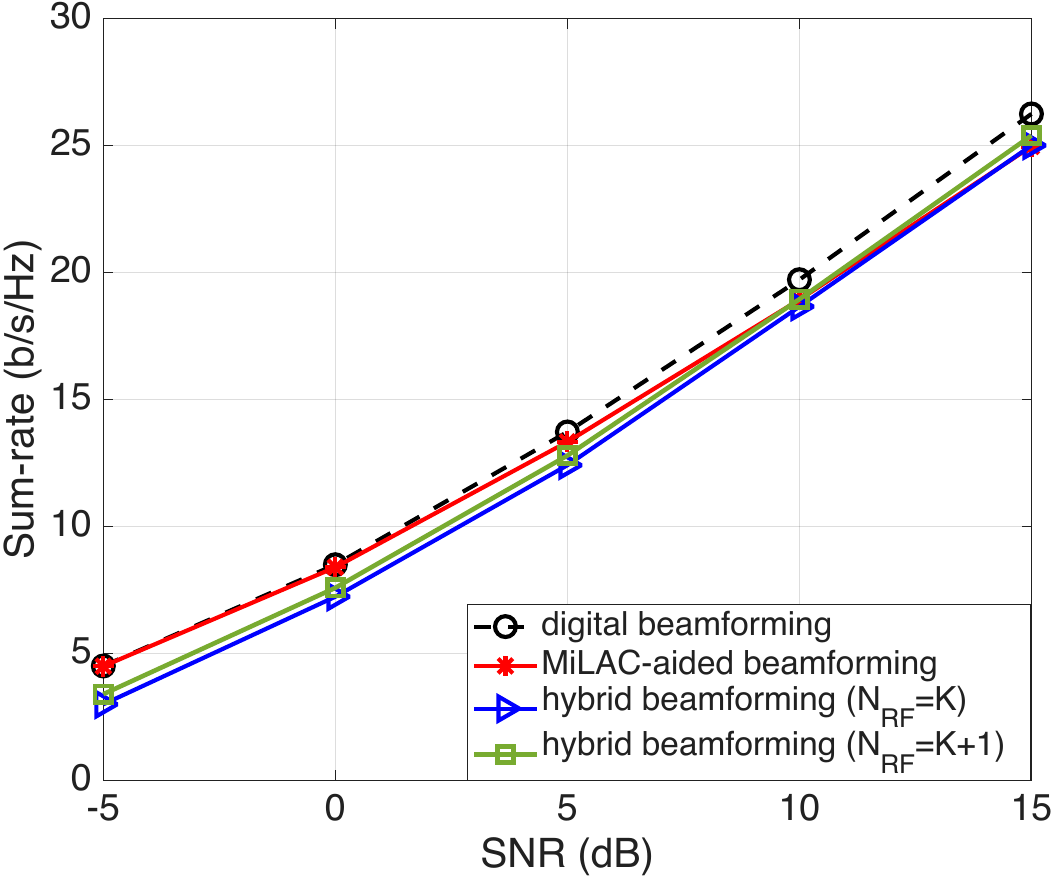}}
\centering
\caption{Sum-rate performance versus SNR for different beamforming schemes, where the result for digital beamforming is obtained by the WMMSE algorithm in \cite{rethinkingwmmse}, the results for hybrid beamforming with $N_{RF}=K$ and $N_{RF}=K+1$ are obtained by approaches in \cite{hybrid1} and \cite{hybrid2}, respectively. The number of transmit antennas and the number of users are $N=64$ and $K=4$, respectively.}
\label{fig:sumrateversussnr}
\end{figure}

Fig. \ref{fig:sumrateversussnr} depicts the sum-rate versus the SNR for different beamforming schemes. For hybrid beamforming, we include the cases with $N_{RF}=K$ and $N_{RF}=K+1$ as baselines to provide a fair comparison with MiLAC-aided beamforming, which requires  $K$ RF chains.
 As shown in the figure, all three beamforming schemes achieve comparable performance. Compared with fully digital beamforming, the MiLAC-aided scheme exhibits a marginal performance degradation. This result is consistent with our conclusion in Section \ref{sec:3b} that MiLAC-aided beamforming cannot offer the full beamforming flexibility of its digital counterpart. Nonetheless, the resulting gap is very small because, under both gigantic  MIMO Rayleigh fading and clustered geometric  channels, the user channel vectors are nearly orthogonal; see Remark~\ref{remark1} and the discussions that follow. Moreover, the performance gap at low SNRs is almost negligible.  This is because at low SNRs, the sum-rate is maximized when full power is allocated to the user with the strongest channel, i.e., only one column of $\bW$ is non-zero. Such a beamforming matrix is achievable by MiLAC-aided beamforming. 
  Compared with hybrid beamforming equipped with $N_{RF}=K$ and $N_{RF}=K+1$ RF chains, MiLAC-aided beamforming attains nearly the same performance in the high-SNR regime, while outperforming both configurations in the low-SNR regime as hybrid beamforming cannot generate arbitrary beamforming vectors due to the limited flexibility offered by phase shifters. 

  
In Fig. \ref{sumrate:wrtN}, we further report the sum-rate of the three beamforming schemes versus the number of transmit antennas. For clarity, we only include hybrid beamforming with $N_{RF}=K$.  In the low-SNR regime (SNR = 0 dB), MiLAC-aided beamforming performs very closely to digital  beamforming and consistently outperforms hybrid beamforming as the number of transmit antennas varies, which further validates the discussions  in the previous paragraph. The performance gap between MiLAC-aided and digital beamforming is larger in the high-SNR regime (SNR = 15 dB), but it gradually shrinks as the number of transmit antennas $N$ increases, since the user channels become more and more orthogonal as $N$ grows. 
 In particular, MiLAC-aided beamforming achieves about $90\%$ of the digital-beamforming sum-rate when $N=16$, and more than $98\%$ when $N=512$ for both channel models.
 
 We can draw the following conclusions regarding digital, MiLAC-aided, and hybrid beamforming. First, MiLAC-aided beamforming cannot achieve exactly the same performance as digital beamforming in multiuser systems. However, the performance gap is very small, especially in the considered gigantic MIMO scenario and in the low-SNR regime. Second, compared with hybrid beamforming, MiLAC-aided beamforming achieves comparable or even superior performance with the same number of RF chains, and has the additional advantages of avoiding symbol-level digital processing, such as per-symbol matrix–vector multiplications, and enabling the use of low-resolution DACs. 
 \begin{figure}
\subfigure[Rayleigh fading channel.]{\includegraphics[width=0.35\textwidth]{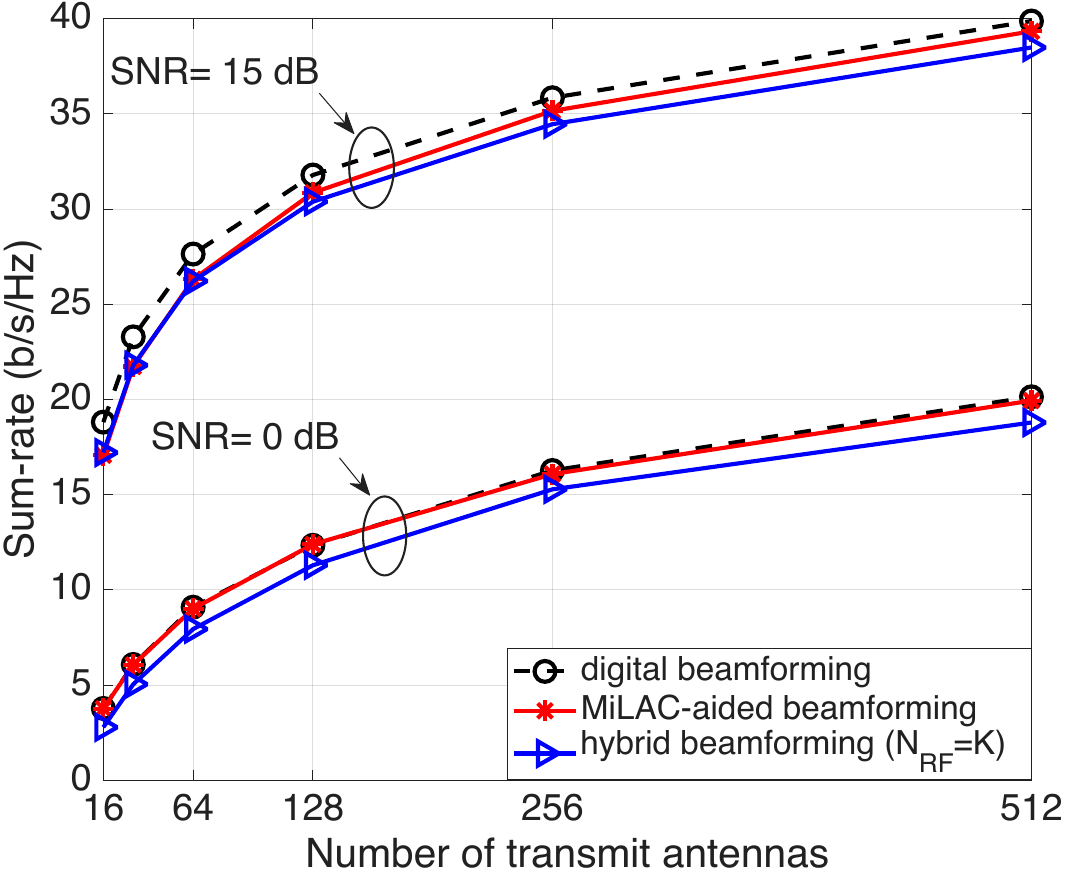}}
\subfigure[Clustered geometric channel.]{\includegraphics[width=0.35\textwidth]{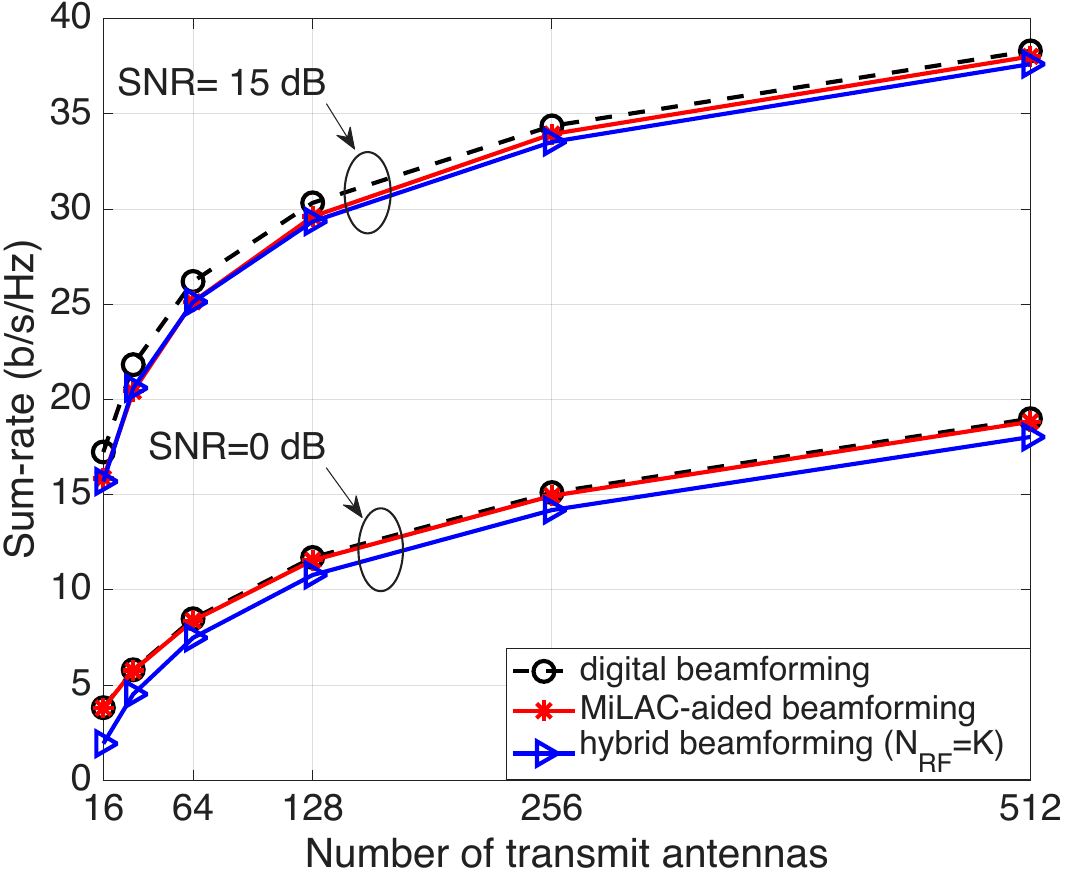}}
\centering
\caption{Sum-rate performance versus $N$ for different beamforming schemes. The number of user is $K=4$.}
\label{sumrate:wrtN}
\end{figure}
\section{Conclusion}\label{sec:conclusion}

This paper investigated the fundamental limits and beamforming design of MiLAC-aided multiuser systems. We provided a complete characterization of the beamforming matrix set achievable by lossless and reciprocal MiLAC, and developed efficient algorithms for MiLAC-aided beamforming design by carefully exploiting the underlying problem structure. Based on the theoretical analysis and simulation results, we draw the following conclusions.

1) MiLAC-aided beamforming does not generally attain the full flexibility of digital beamforming. However, it achieves comparable performance in gigantic MIMO systems.

2) Compared with hybrid beamforming, MiLAC-aided beamforming achieves comparable or superior performance with the same number of RF chains, while additionally avoiding symbol-level digital processing and enabling the use of low-resolution DACs, thereby reducing hardware and computational complexity.

3) MiLAC-based analog beamforming offers greater flexibility than conventional phase-shifter-based analog beamforming, and thus a hybrid digital–MiLAC architecture provides higher flexibility than conventional hybrid beamforming. In particular, it achieves full digital beamforming flexibility with $K$ RF chains, halving that required by conventional hybrid beamforming. 
\appendices
\section{Proof of Proposition \ref{sp_eq}}\label{app:proof:sp_eq}
We first write the KKT conditions of problems \eqref{sumrate2} and \eqref{sumrate3}.  The KKT condition of problem \eqref{sumrate2} is
\begin{subequations}
\begin{align}
&\nabla_{\bW} R(\bW)-\bW\boldsymbol{\Lambda}=\mathbf{0},\label{9a}\\
&{\Lambda}_{k,k}-\mu=0,~k=1,2,\dots, K,\\
&\left<\boldsymbol{\Lambda},\text{diag}(\mathbf{p})-\bW^H\bW\right>=0,~~\boldsymbol{\Lambda}\succeq\mathbf{0},~\bW^H\bW\preceq \text{diag}(\mathbf{p}),\label{9c}\\
&\mu(P_T-\mathbf{1}^T\mathbf{p})=0,~~\mu\geq 0,~\mathbf{1}^T\mathbf{p}\leq P_T,
\end{align}
\end{subequations}
where $R(\bW)=\sum_{k=1}^K\log\left(1+\frac{|\h_k^H\mathbf{w}_{k}|^2}{\sum_{j\neq k}|\h_k^H\bw_{j}|^2+\sigma^2}\right)$ is the sum-rate objective function, $\boldsymbol{\Lambda}$ and $\mu$ are the Lagrange multipliers associated with constraints $\bW^H\bW\preceq\text{diag}(\mathbf{p})$ and $\mathbf{1}^T\mathbf{p}\leq P_T$, respectively.

 Similarly, the KKT condition of problem \eqref{sumrate3} is
\begin{subequations}
\begin{align}
&\bH\nabla_{\bW} R(\bH^H\bX)-\bar{\bH}\bX\boldsymbol{\Lambda}=\mathbf{0},\label{10a}\\
&{\Lambda}_{k,k}-\mu=0,~k=1,2,\dots, K,\label{10b}\\
&\left<\boldsymbol{\Lambda},\text{diag}(\mathbf{p})-\bX^H\bar{\bH}\bX\right>=0,~~\boldsymbol{\Lambda}\succeq\mathbf{0},~\bX^H\bar{\bH}\bX\preceq \text{diag}(\mathbf{p}),\label{10c}\\
&\mu(P_T-\mathbf{1}^T\mathbf{p})=0,~~\mu\geq 0,~\mathbf{1}^T\mathbf{p}\leq P_T.\label{10d}
\end{align}
\end{subequations}
Let $\bZ=\bH\bW$, the sum-rate can be expressed as   
$$
\begin{aligned}
R(\bZ(\bW))=&\log(\|[\bZ]_{k,:}\|_2^2+\sigma^2)\\
&-\log(\|[\bZ]_{k,:}\|_2^2-|Z_{k,k}|^2+\sigma^2).
\end{aligned}$$
 Utilizing the relation that $\nabla_{\bW}R(\bW)=\bH^H\nabla_\bZ R(\bZ)$ and by simple calculation, the gradient $\nabla_{\bW}R(\bW)$ can be expressed as 
$$\nabla_{\bW}R(\bW)=\bH^H\bG(\bH\bW),$$ 
where the mapping $\bG(\cdot): \C^{K\times K}\rightarrow \C^{K\times K}$ is the gradient of $R$ with respect to $\bZ$, given by
$$
\begin{aligned}
[\bG(\bZ)]_{k,:}&=\frac{\bZ_{k,:}}{\log(\|\bZ_{k,:}\|_2^2+\sigma^2)}\\
&-\frac{\bZ_{k,:}-Z_{k,k}\mathbf{e}_k^T}{\log(\|\bZ_{k,:}\|_2^2+\sigma^2-|Z_{k,k}|^2)},~k=1,2,\dots, K.
\end{aligned}$$

Let $(\bW,\mathbf{p},\boldsymbol{\Lambda},\mu)$ be a KKT point of \eqref{sumrate2}, and let $\bX=\bar{\bH}^{-1}\bH\bW$. We next prove that  $(\bX,\mathbf{p},\boldsymbol{\Lambda},\mu)$ is a KKT point of \eqref{sumrate3}. It suffices to prove that \eqref{10a} and \eqref{10c} are satisfied. Clearly, $\Pi_{\mathbf{H}^H}\mathbf{W}=\bH^H\bX. $ Let 
$$\widetilde{\bW}=(\mathbf{I} - \Pi_{\mathbf{H}^H})\bW=\bW-\bH^H\bX\in\ker(\bH).$$
From \eqref{9a}, we have  
$$\nabla_{\bW}R(\bW)-\bW\boldsymbol{\Lambda}=\bH^H\bG(\bH\bW)-\bW\boldsymbol{\Lambda}=\mathbf{0}.$$ Substituting $\bW=\bH^H\bX+\widetilde{\bW}$  and using $\bH\widetilde{\bW}=\mathbf{0}$ yields
$$\widetilde{\bW}\boldsymbol{\Lambda}=\bH^H\bG(\bH\bH^H\bX)-\bH^H\bX\boldsymbol{\Lambda}.$$
Therefore, $\widetilde{\bW}\boldsymbol{\Lambda}\in\ker(\bH)\cap\text{Ran}(\bH^H)$ and is thus a zero matrix.  The condition \eqref{10a} follows immediately.  In addition,  
$$
\begin{aligned}
&\left<\boldsymbol{\Lambda},\text{diag}(\mathbf{p})-\bX^H\bar{\bH}\bX\right>\\
&\overset{(a)}{=}\left<\boldsymbol{\Lambda},\text{diag}(\mathbf{p})-\bW^H\bW\right>+\left<\boldsymbol{\Lambda},\widetilde{\bW}^H\widetilde{\bW}\right>\overset{(b)}{=}0,
\end{aligned}$$
and 
$$\bX^H\bar{\bH}\bX\overset{(c)}{=}\bW^H\bW-\widetilde{\bW}^H\widetilde{\bW}\overset{(d)}{\preceq}\text{diag}(\mathbf{p}),$$
where (a) and (c) apply $\bW=\bH^H\bX+\widetilde{\bW}$, (b) follows from \eqref{9c} and $\widetilde{\bW}\boldsymbol{\Lambda}=\mathbf{0}$, and (d) follows from \eqref{9c}. Therefore, \eqref{10c} is satisfied, which completes our proof for the first result. 

Now let  $(\bX,\mathbf{p},\boldsymbol{\Lambda},\mu)$ be a KKT point of problem \eqref{sumrate3}. We show that $(\bH^H\bX,\mathbf{p},
\boldsymbol{\Lambda},\mu)$ is a KKT point of problem \eqref{sumrate2}. In particular, we only need to show that \eqref{9a} is satisfied.  From \eqref{10a},
$$\bar{\bH}\bG(\bH\bH^H\bX)-\bar{\bH}\bX\boldsymbol{\Lambda}=\mathbf{0}.$$
Multiplying by $\bH^H\bar{\bH}^{-1}$ from the l.h.s. of the above equation gives \eqref{9a}. This completes the proof. 
 \section{Proof of Theorems \ref{converge1} and \ref{converge2}}\label{app:convergence}
Algorithms \ref{WMMSE} and \ref{WMMSE-LC} are inherently BCD algorithms applied to their corresponding optimization problems \eqref{problem:wmmse} and \eqref{problem:wmmse2}, respectively. For both algorithms, the $\boldsymbol{\omega}$- and $\mathbf{u}$-subproblems are strictly convex and thus admit unique solutions. For Algorithm \ref{WMMSE-LC}, the $\mathbf{p}$-subproblem in \eqref{p-subproblem} also has as a unique solution, given in \eqref{update:pk}, since $\lambda\geq 0$ is uniquely determined. In addition, the iterations generated by both algorithms lie in a bounded set. This is because the feasible region of $(\bX,\mathbf{p})$ in problem \eqref{problem:wmmse} (or $(\mathbf{Y},\mathbf{p})$ in problem \eqref{problem:wmmse2}) is compact, which, together with the update rules in \eqref{uk} and \eqref{omegak}, further implies that the iterates of $(\boldsymbol{\omega},\mathbf{u})$ lie in a compact set. In particular, \eqref{omegak} can be rewritten as $$\omega_k=1+\frac{|\bar{\h}_k^H\x_k|^2}{\sum_{j\neq k}|\bar{\h}_k^H\x_j|^2+\sigma^2},$$ which is bounded. Combining the above, we can conclude that both algorithms satisfy: (i) the objective function is differentiable; (ii) at least $n-1$ of the $n$ block subproblems admit unique solutions (the $\boldsymbol{\omega}$ and $\mathbf{u}$ blocks for Algorithm \ref{WMMSE} and the $\boldsymbol{\omega}$, $\mathbf{u}$, and $\mathbf{p}$ blocks for Algorithm \ref{WMMSE-LC});  and (iii) the generated iterations lie in a bounded set. Hence, according to \cite[Theorem 2 (b)]{BCD}, any limit point generated by the algorithms converges to a stationary point.  The second assertion of the theorems can be proved using the same arguments as in \cite[Theorem 3]{wmmse}.
 
   \bibliographystyle{IEEEtran}
\bibliography{IEEEabrv,milac}

 \end{document}